\documentclass[letterpaper,11pt]{article}
\usepackage[margin=1in]{geometry}
\usepackage{amsmath,amsfonts,amsthm,amssymb,bm, verbatim,dsfont,mathtools}
\usepackage{color,graphicx,appendix}
\usepackage{etoolbox}
\usepackage{tikz}
\usepackage{xr,xspace}
\usepackage{todonotes}
\usepackage{paralist}
\usepackage{caption,soul}
\usepackage[ruled,vlined]{algorithm2e}
\usepackage{microtype}

\newtheorem{proposition}{Proposition}

\usepackage{xspace,prettyref}

\newcommand{\diverge}{\to\infty}

\newcommand{\integers}{{\mathbb{Z}}}

\newcommand{\expect}[1]{\mathbb{E}\left[ #1 \right]}

\newcommand{\Bern}{{\rm Bern}}

\newcommand{\iid}{i.i.d.\xspace}
\newrefformat{eq}{(\ref{#1})}
\newrefformat{chap}{Chapter~\ref{#1}}
\newrefformat{sec}{Section~\ref{#1}}
\newrefformat{alg}{Algorithm~\ref{#1}}
\newrefformat{fig}{Fig.~\ref{#1}}
\newrefformat{tab}{Table~\ref{#1}}
\newrefformat{rmk}{Remark~\ref{#1}}
\newrefformat{clm}{Claim~\ref{#1}}
\newrefformat{def}{Definition~\ref{#1}}
\newrefformat{cor}{Corollary~\ref{#1}}
\newrefformat{lmm}{Lemma~\ref{#1}}
\newrefformat{prop}{Proposition~\ref{#1}}
\newrefformat{app}{Appendix~\ref{#1}}
\newrefformat{hyp}{Hypothesis~\ref{#1}}
\newrefformat{thm}{Theorem~\ref{#1}}
\newrefformat{ass}{Assumption~\ref{#1}}

\newcommand{\pth}[1]{\left( #1 \right)}

\newcommand{\sth}[1]{\left\{ #1 \right\}}

\newcommand{\abth}[1]{\left | #1 \right |}

\newcommand{\indc}[1]{{\mathbf{1}_{\left\{{#1}\right\}}}}

\newcommand{\calA}{{\mathcal{A}}}
\newcommand{\calB}{{\mathcal{B}}}

\newcommand{\calE}{{\mathcal{E}}}

\newcommand{\calS}{{\mathcal{S}}}

\renewcommand{\tilde}{\widetilde}

\theoremstyle{plain}
\newtheorem{theorem}{Theorem}[section]
\newtheorem{lemma}[theorem]{Lemma}
\newtheorem{corollary}[theorem]{Corollary}

\theoremstyle{definition}

\newtheorem{example}{Example}

\theoremstyle{remark}
\newtheorem{remark}[theorem]{Remark}

\newcommand{\ignore}[1]{}

\begin{document}
\title{
Collaboratively Learning the Best Option on Graphs, Using Bounded Local Memory
}
\author{Lili Su, Martin Zubeldia, Nancy Lynch 
\thanks{This work was initialized by L. Su and N. Lynch. L. Su and N. Lynch formulated the problem, 
developed the methods. The analysis consists of two parts, the learnability (Section 4.1) and transient system 
behaviors (Section 4.2).  
The analysis in Section 4.1 was done by L. Su, and was commented by M. Zubeldia and N. Lynch. 
The analysis in Section 4.2 was done by M. Zubeldia, and was commented by L. Su. 
Theorem 4.8 and Corollary 4.9 in Section 4.2 was done by L. Su. 
Simulation was done by M. Zubeldia. 
The manuscript was written by L. Su, revised by M. Zubeldia, and commented by N. Lynch. Section 4.3 was added by M. Zubeldia.}\\
~
~ Department of Electrical Engineering and Computer Science\\
Massachusetts Institute of Technology ~ 
}
\date{}

\maketitle

\begin{abstract}

We consider multi-armed bandit problems in social groups wherein each individual has bounded memory 
and shares the common goal of learning the best arm/option. 
We say an individual learns the best option if eventually (as $t\diverge$) it pulls only the arm with the highest expected reward.
While this goal is provably impossible for an isolated individual due to bounded memory, 
we show that, in social groups, this goal can be achieved easily with the aid of social persuasion (i.e., communication) as long as the communication networks/graphs satisfy some mild conditions.
In this work, we model and analyze a type of learning dynamics which are well-observed in social groups. 
Specifically, under the learning dynamics of interest, an individual sequentially decides on which arm to pull next based on not only its private reward feedback but also the suggestion provided by a randomly chosen neighbor.
To deal with the interplay between the randomness in the rewards and in the social interaction, we employ the {\em mean-field approximation} method.
Considering the possibility that the individuals in the networks may not be exchangeable when the communication networks are not cliques, we go beyond the classic mean-field techniques and apply a refined version of mean-field approximation:  
\begin{itemize}
	\item  Using coupling 
	we show that, if the communication graph is connected and is either regular or has doubly-stochastic degree-weighted adjacency matrix, with probability $\to 1$ as the social group size $N \diverge $, every individual in the social group learns the best option.

	\item If the minimum degree of the graph diverges as $N \diverge $, over an arbitrary but given finite time horizon, the sample paths describing the opinion evolutions of the individuals are asymptotically independent.  In addition, the proportions of the population with different opinions converge to the unique solution of a system of ODEs. Interestingly, the obtained system of ODEs are invariant to the structures of the communication graphs. In the solution of the obtained ODEs, the proportion of the population holding the correct opinion converges to $1$ exponentially fast in time.
	 \end{itemize}
Notably, our results hold even if the communication graphs are highly sparse.
\end{abstract}

\section{Introduction}
\label{sec: intro}
Individuals often need to make a sequence of decisions among a fixed finite set of options (alternatives), whose rewards/payoffs can be regarded as stochastic, for example:
\begin{itemize}
\item Human society: In many economic situations, individuals need to make a sequence of decisions among multiple options, such as when purchasing perishable products \cite{10.2307/2171959} and when designing financial portfolios \cite{shen2015portfolio}. In the former case, the options can be the product of the same kind from different sellers. In the latter, the options are different possible portfolios.
\item Social insect colonies and swarm robotics: Foraging and house-hunting are two fundamental problems in social insect colonies, and both of them have inspired counterpart algorithms in swarm robotics \cite{pini2012multi}. During foraging, each ant/bee repeatedly refines its foraging areas to improve harvesting efficiency. House-hunting refers to the collective decision process in which the entire social group collectively identifies a high-quality site to immigrate to. For the success of house-hunting, individuals repeatedly scout and evaluate multiple candidate sites, and exchange information with each other to reach a collective decision.
\end{itemize}
Many of these sequential decision problems can be cast as {\em multi-armed bandit problems} \cite{lai1985asymptotically,auer2002finite,bubeck2012regret}. These have been studied intensively in the centralized setting, where there is only one player in the system, under different notions of performance metrics such as pseudo-regret, expected regret, simple regret, etc. \cite{lai1985asymptotically,auer2002finite,bubeck2012regret,mannor2004sample,robbins1956sequential,bubeck2012regret}. Specifically, a $K$-armed bandit problem is defined by the reward processes of individual arms/options $\pth{R_{k, i}: i \in \integers_+} $ for $k=1, \cdots, K$, where $R_{k, i}$ is the reward of the $i$--th pull of arm $k$.
At each stage, a player chooses one arm to pull and obtains some observable payoff/reward generated by the chosen arm. In the most basic formulation the reward process $\pth{R_{k, i}: i\in \integers_+}$ of each option is stochastic and successive pulls of arm $k$ yield $\iid$ rewards $R_{k, 1}, R_{k, 2}, \cdots$. 
Both asymptotically optimal algorithms and efficient finite-time order optimal algorithms 
have been proposed \cite{robbins1956sequential,auer2002finite,bubeck2012regret,shahrampour2017sequential}. These algorithms typically have some non-trivial requirements on individuals' memorization capabilities. 
For example, upper confidence bound (UCB) algorithm requires an individual to memorize the cumulative rewards of each arm he has obtained so far, the number of pulls of each arm, and the total number of pulls \cite{robbins1956sequential,auer2002finite}. Although this is not a memory-demanding requirement,
nevertheless, this requirement cannot be perfectly fulfilled even by humans, let alone by social insects, due to bounded rationality of humans, and limited memory and inaccurate computation of social insects. In human society, when a customer is making a purchase decision of perishable products, he may recall only the brand of product that he is satisfied with in his most recent purchase. Similarly, in ant colonies, during house-hunting, an ant can memorize only a few recently visited sites.

In this paper, we capture the above memory constraints by assuming an individual has only bounded/finite memory. The problem of multi-armed bandits with {\em finite memory constraint} has been proposed by Robbins \cite{robbins1956sequential} and attracted some research attention \cite{smith1965robbins,cover1968note,1054427}. 
The subtleties and pitfalls in making a good definition of memory were not identified until Cover's work \cite{1054427,cover1968note}. We use the memory assumptions specified in \cite{1054427}, which require that an individual's history  be summarized by a finite-valued memory. The detailed description of this notion of memory can be found in Section \ref{sec: model}. We say an individual learns the best option if eventually (as $t\diverge$) it pulls only the arm with the highest expected reward.

For an isolated individual, learning the best option is provably impossible \cite{1054427}.\footnote{A less restricted memory constraint -- stochastic fading memory -- is considered in \cite{xureinforcement}, wherein similar negative results when memory decays fast are obtained.}
Nevertheless, successful learning is still often observed in social groups such as human society \cite{10.2307/2171959}, social insect colonies \cite{nakayama2017nash} and swarm robotics \cite{pini2012multi}. This may be because in social groups individuals inevitably interact with others. In particular, in social groups individuals are able to, and tend to, take advantage of others' experience through observing their neighbors \cite{bandura1969social,rendell2010copy}. Intuitively, it appears that as a result of this social interaction, the memory of each individual is ``amplified'', and this {\em amplified shared memory} is sufficient for the entire social group to collaboratively learn the best option.

\vskip \baselineskip

\noindent{\bf Approach and key contributions}:
In this paper, we rigorously show that the above intuition is correct with a focus on the impact of the graph structures on the performance of collaboratively learning.
We study the learning dynamics wherein an individual makes its local sequential decisions on which arm to pull next based on not only its private reward feedback but also the suggestions provided by randomly chosen neighbors. Concretely, we assume time is continuous and each individual has an independent Poisson clock with common rate. The  Poisson clocks model is very natural and has been widely used \cite{ross2014introduction,shwartz1995large,kurtz1981approximation,hajek2015random}.  
When an individual's local clock ticks, it attempts to perform an update immediately via two steps:
\begin{enumerate}
	\item {\bf Sampling}:
	 {\bf If} the individual does not have any preference over the $K$ arms yet, {\bf then}:
	\begin{enumerate}
		\item With probability $\mu\in [0,1]$, the individual pulls one of the $K$ arms uniformly at random (uniform sampling).
		\item With probability $1-\mu$, the individual chooses one neighbor uniformly at random, and pulls the arm suggested by the chosen neighbor (peer recommendation);  
		pulls no arm if the chosen neighbor does not have any preference over the $K$ arms yet.
	\end{enumerate}
	{\bf else}
	The individual chooses one neighbor uniformly at random, and pulls the arm suggested
	 by the chosen neighbor (peer recommendation); pulls no arm if the chosen neighbor does not have any preference over the $K$ arms yet.
	\item {\bf Adopting}: {\bf If} the stochastic reward generated by the pulled arm is 1, {\bf then} the individual updates its preference to this arm.
\end{enumerate}
Note that if the awake individual pulls no arm, it will not get a reward 1; thus, its preference is unchanged.
Formal description can be found in Section \ref{sec: model}. Our learning dynamics are similar to those studied in \cite{Celis:2017:DLD:3087801.3087820} but with the following key differences: (1) We consider general communication graphs; in contrast,  only cliques are considered in \cite{Celis:2017:DLD:3087801.3087820}. (2) We consider continuous-time and asynchronous  dynamics, whereas in \cite{Celis:2017:DLD:3087801.3087820} all individuals are required to make update simultaneously. (3) Under our learning rules, with high probability, every individual learns the best option; whereas in \cite{Celis:2017:DLD:3087801.3087820}, the proportions of the population with the wrong opinions are bounded away from zero as long as $\mu>0$.
These differences are fundamental and require completely new analysis.

A key analytical challenge of our learning dynamics is to deal with the interplay of the randomness in the rewards and that in the social interaction. Comparing to the case when the communication graphs are cliques, this interplay is significantly complicated by the lack of exchangeability among the individuals on general communication graphs.  Observing this, we go beyond the classic mean-field techniques and apply a refined version of mean-field approximation:  
\begin{itemize}
	\item  We show that if the communication graph is connected and is either regular or has doubly-stochastic degree-weighted adjacency matrix, with probability $\to 1$ as the social group size $N \diverge $, every individual in the social group learns the best option with local memory of $O(K)$ states.

For any fixed $N$, we couple the evolutions of the number of the individuals with the correct opinion with a standard biased random walk, whose success probability is well-understood.
	\item If the minimum degree of the graph diverges as $N \diverge $, over an arbitrary but given finite time horizon,  the sample paths describing the opinion evolutions of the individuals are asymptotically independent. In addition, the proportions of the population with different opinions converge to the unique solution of a system of ODEs. 
	Interestingly, the obtained system of ODEs are invariant to the structures of the communication graphs. In the solution of the obtained ODEs, the proportion of the population holding the correct opinion converges to $1$ exponentially fast in time.
%
	
	The key challenge in the analysis of general graphs is that due to the lack of exchangeability, one needs to keep track of the opinion evolution of each individual. This complicates the state description of the system. 
	 \end{itemize}
Notably, our results hold even if the communication graphs are highly sparse.

It is easy to see that the time needed for the {\em entire} social group to learn the best option scales poorly in $N$ -- as the entire social group cannot learn the best option until every individual wakes up at least once.
%
Fortunately, it turns out that in many applications such as social insect colonies, it suffices to know the convergence rate until a sufficiently large fraction of the population have the correct opinion; this can be obtained by exploring the transient system behaviors.

\section{Related Work}
\begin{itemize}

\item {\bf Multi-armed bandits with bounded memory: }
Multi-armed bandits with finite memory has been proposed by Robbins \cite{robbins1956sequential} in the special setting where the bandit has only two arms. 
The goal there is to maximize the long-run proportion of heads obtained which is more relaxed than to identify the best arm asymptotically. 
Efforts have been made to extend and improve the results in the seminal work \cite{robbins1956sequential}. Cover \cite{cover1968note} constructed a time-dependent deterministic allocation rule with two arms, which is shown to be asymptotically uniformly best among the class of time-dependent finite memory rules. However, the implicit assumption is that  
keeping track of time is costless.  The subtleties and pitfalls in making a good definition of memory was not notified until Cover's work \cite{cover1968note}.
The finite-memory constraint defined by Cover and Hellman in the work \cite{1054427} is more relevant to our setting.
It is shown in \cite{1054427} that the optimal value of the long-run proportion over all local update functions and all allocation rules is bounded away from one, and approaches one only with local memory size goes to infinity.

	\item {\bf Plurality consensus: } Another line of work that is closely to our work is plurality consensus \cite{kempe2003gossip,ghaffari2016polylogarithmic,becchetti2015,becchetti2017simple,draief2012convergence}. Plurality consensus problem considers $N$ 
	agents each initially supporting an option in $\sth{1, 2, \cdots, K}$, and the goal is to learn, in a distributed fashion, the option with the largest initial support. 
	The key differences between our problem and the plurality consensus problem are: (1) In our problem, the information on the arm qualities arrives at the system on the fly; whereas plurality consensus does not consider this new information. (2) In our problem, the two sources of randomness -- the randomness in the rewards and the randomness in social interaction -- interplay with each other, and this interplay and associated dependency do not appear in the plurality consensus problem. (3) Plurality consensus relies crucially on the initial configuration of the system; whereas in our problem, the system converges to the right absorbing state even if the best arm is not the most popular one initially.

	\item  {\bf Mean-field approximation: }
Mean-field approximation has been a powerful tool for studying the behavior of large and complex stochastic models for a long time \cite{kurtz1970solutions}. There are many different variants of mean-field models; thus, share some non-trivial similarity \cite{xu2013supermarket,TX12,GTZ17,bramson1998state,wormald1995differential}.
%
In contrast to classic mean-field approximation where the nodes are exchangeable, as is the case with cliques,
to the best of our knowledge, mean-field approximation on general networks/graphs is less well-understood. Recent years have witness a flurry of research on the mean-field approximation for queueing networks with general graphs. In the context of resource pooling, Tsitsiklis and Xu \cite{tsitsiklis2017flexible} considered the scenario where the servers and queues are connected through a bipartite graph. In the context of load balancing, Mukherjee et al. \cite{mukherjee2018asymptotically}  studied the join-the-shortest queue (JSQ) policy on general graphs, and showed that this policy is optimal in some sense, as long as the expected degree diverges fast enough as a function of the number of servers. Similarly, Budhiraja et al. \cite{budhiraja2017supermarket} studied the local {\em power-of-$d$} scheme on graphs, and showed that as long as the minimum degree of the graph diverges, and the degree-weighted adjacency matrix of the graph is asymptotically doubly-stochastic, the occupancy process converges to the same system of ODEs as that on the complete graph.

\item {\bf Propagation of chaos: }
In general, a powerful stepping stone to establishing mean-field approximation results when the network/graph is not symmetric, is the asymptotic independence of the local state of the agents in the graph (this is also called ``propagation of chaos''). This has been done in the context of weakly interacting particles \cite{Bhamidi2018}, queueing systems  \cite{budhiraja2017supermarket,asymptoticIndependence,delayAsymptotics},
and in other general systems \cite{chaosHypothesis,propagationOfChaos}.
	
\end{itemize}

\section{Model}
\label{sec: model}
In this section, we formalize the social group, the learning goal, and the notion of bounded memory,  describe the learning dynamics that we are interested in, and define a continuous-time Markov chain. 

\subsubsection*{\bf Social group}

A social group consists of $N$ homogeneous individuals/agents that are connected through an undirected graph $G_N = \pth{[N], \calE^N}$, where $[N]:= \sth{1, \cdots, N}$ and $\calE^N$ is a collection of edges among $[N]$. Each agent $i \in [N]$ has a set of neighbors $V_i^N$. Let
\begin{align}
\label{eq: min degree}
D_{\min}^N := \min_{i\in[N]} \left|V_i^N\right|
\end{align}
be the minimum degree of $G_N$.
Let $\bm{D}^N$ be the diagonal matrix of degrees, and $\bm{A}^N$ be the adjacency matrix.
With a little abuse of terminology, we say graph $G_N$ is doubly-stochastic if the degree-weighted adjacency matrix
\begin{align*}
(\bm{D}^N)^{-1} \bm{A}^N
\end{align*}
is doubly-stochastic. Equivalently, graph $G_N$ is doubly-stochastic if
\[
\sum_{j \in V_i^N} \frac{1}{|V_j^N|} = \sum_{j \in V_i^N} \frac{1}{|V_i^N|} =1, ~ ~ ~ ~ \forall\, i\in[N].
\]

\subsubsection*{\bf Learning goal}

The agents in the social group want to collaboratively solve the $K$-armed stochastic bandit problems, wherein the reward processes of the $K$ arms/options are Bernoulli processes with parameters $p_1, \cdots, p_K$. 
If arm $a_k$ is pulled at time $t$, then reward $R_t \sim \Bern\pth{p_k}$, i.e.,
\begin{align*}
R_{t} =
\begin{cases}
1, & \text{with probability }p_k; \\
0, & \text{otherwise.}
\end{cases}
\end{align*}
Initially the distribution parameters $p_1, \cdots, p_K$ are unknown to any agent. We assume the arm with the highest parameter $p_k$ is unique. Without loss of generality, let $a_1$ be the unique best arm and  $p_1>p_2\geq \cdots p_K\geq 0$. 
Each agent has an independent Poisson clock\footnote{The Poisson clocks model is very natural and has been widely adopted for modeling \cite{ross2014introduction,shwartz1995large,kurtz1981approximation,hajek2015random} natural dynamics.  } with common rate $\lambda$, and attempts to pull an arm immediately when its local clock ticks.
We say an agent learns the best option if, as $t\diverge$, it pulls only the arm with the highest expected reward, i.e., $a_1$.

\subsubsection*{\bf Bounded memory}

We assume that each agent has finite/bounded memory \cite{1054427}. 
Measuring the memory size in terms of states, as an alternative to bits, is rather standard in many natural dynamics \cite{1054427,ghaffari2016polylogarithmic}.
We say an agent has a memory of size $m$ if its experience is completely summarized by an $m$-valued variable $M\in \sth{0, \cdots, m-1}$.
As a result of this, an agent sequentially decides on which arm to pull next based on only (i) its memory state and (ii) the information it gets through social interaction. The memory state may be updated with the restriction that only (a) the current memory state, (b) the current choice of arm, and (c) the recently obtained reward, are used for determining the new state. In the learning dynamics considered in this paper, it suffices to have $m=K+1$.

\subsubsection*{\bf Learning dynamics}

In the learning dynamics under consideration, each agent keeps two variables:
\begin{itemize}
	\item a local memory variable $M$ that takes values in $\sth{0, 1, \cdots, K}$. If $M=0$, the agent does not have any preference over the $K$ arms; if $M=k \in \{1, \cdots, K\}$, it means that tentatively the agent prefers arm $a_k$ over others.
	\item an arm choice variable $c$ that takes values in $\sth{0, 1, \cdots, K}$ as well. If $c=0$, the agent pulls no arm; if $c=k \in \{1, \cdots, K\}$, the agent chooses arm $a_k$ to pull next.
\end{itemize}
Note that $M$ is the persistent memory while $c$ is a temporary variable. There are a variety of choices in initializing $M$, either deterministically or randomly. Our results account for different choices of initialization.  
%
For the temporary variable $c$, initialize $c=M$. 


When the clock at agent $i$ ticks at time $t$, we say agent $i$ obtains the memory refinement token. With such a token, agent $i$ refines its memory $M$ according to Algorithm \ref{alg: 1} via a two-step procedure. 
The first two {\bf if}--{\bf else} clauses describe how to choose an arm to pull next: 
If an agent does not have any preference ( i.e., $M=0$), $c$ is determined through a combination of uniform sampling and peer recommendation; otherwise, $c$ is completely determined by peer recommendation, see Section \ref{sec: intro} for the notions of {\em uniform sampling} and {\em peer recommendation}. Note that during peer recommendation, if the chosen peer does not have any preference over the $K$ arms yet, the memory of the awake agent remains unchanged.  The last {\bf if}--{\bf else} clause says that as long as the reward obtained by pulling the arm is 1 (i.e., $R_t=1$), then $M\gets c$; otherwise, $M$ is unchanged.

\begin{algorithm}

\caption{Collaborative Best Option Learning (at agent $i$)}
\label{alg: 1}
{\bf Input}:  $K$, $\mu$, and $V_i^N$\;    
~ {\em Local variables}: $M \in \sth{0, 1, \cdots, K}$  and $c\in \sth{0, 1, \cdots, K}$\;
~ 
Initialize $c=M$\;

\vskip 0.2\baselineskip

\SetKwFunction{FSO}{PeerRecommendation}

When local clock ticks: \\

\eIf{$M=0$}
{
	With probability $\mu$, set $c$ to be one of the $K$ arms uniformly at random\;

    With probability $1-\mu$, $c \gets$ \FSO \;  
}
{$c \gets$  \FSO \;}

\eIf{$c=0$}
{Pull no arm\;}
{Pull arm $c$\;}

\vskip 0.8\baselineskip

\eIf{$R_t=1$}{$M \gets c$;}
{$M$ unchanged\; }

\vskip 0.6\baselineskip
\FSO{}
{ \\
\vskip 0.6\baselineskip
Choose one neighbor $i^{\prime}$ 
from $V_i^{N} $
uniformly at random\;

\KwRet $M^{\prime}$;     ~~~~~  \%\% $M^{\prime}$ is the memory state of $i^{\prime}$\;

}
\end{algorithm}
%
Similar gossip model was used in \cite{ghaffari2016polylogarithmic}. 
In our model, different from \cite{ghaffari2016polylogarithmic}, the
reward information arrives at the system on the fly.

\subsubsection*{\bf System state}

Due to the lack of exchangeability among agents, one needs to keep track of the opinion evolution of each agent. 
For a given $N$, let $\bm{X}^N(t)\in \{0,1\}^{N \times (K+1)}$ denote the state of the system at time $t$ defined as 
\begin{align}
\label{eq:dl mc 1}
\bm{X}^N(t) ~ : = \pth{\bm{X}_1^N (t), \cdots, \bm{X}_N^N (t)},
\end{align}
where 
\begin{align}
\label{eq:dl mc}
\bm{X}_i^N (t) = \Big(\bm{X}_{i,0}^N (t), \bm{X}_{i,1}^N (t), \cdots, \bm{X}_{i,K}^N (t)\Big)^{\top} \in \{0,1\}^{K+1}
\end{align}
represents agent $i$'s opinion at time $t$ with
\begin{align*}
\bm{X}_{i,k}^N (t)=1, ~ and ~ ~ \bm{X}_{i,k^{\prime}}^N (t)=0, \forall k^{\prime}\not=k
\end{align*}
if $M=k$ at agent $i$ at time $t$, i.e., the memory state of agent $i$ at time $t$ is $k$.
Intuitively, if $\bm{X}_{i,k}^N (t)=1$, we say agent $i$ prefers arm $a_k$ at time $t$. Here the vectors under discussion are column vectors. Note that each state $x$ is a $N\times (K+1)$ matrix of entries either 0 or 1, and for each row $i\in [N]$, the entries sum up to 1.
%
Throughout the process, $\bm{X}_{i,k}^N (t)$ is ``fluctuating'' between 0 and 1 as the local clock ticks.

Under Algorithm \ref{alg: 1}, the evolution of the state $\bm{X}^N (\cdot)$ is a continuous-time Markov chain with the following transition rates. For all $i\in[N]$, and for all $k=1,\dots,K$, we have
\begin{align*}
  q_{x, x+e_{ik}-e_{i0}} &= \lambda p_k \left( \frac{\mu}{K} + \left(1-\mu\right) \frac{1}{|V_i^N|} \sum\limits_{\ell\in V_i^N} \indc{x_{\ell, k} =1} \right) \indc{x_{i, 0} =1}, \\
  q_{x, x+e_{ik}-e_{ij}} &= \lambda p_k \left( \frac{1}{|V_i^N|} \sum\limits_{\ell\in V_i^N} \indc{x_{\ell, k} =1} \right) \indc{x_{i, j} =1}, \quad \forall \, j\neq k,
\end{align*}
where $e_{ik}$ is the unit matrix with a single one in position $(i,k)$, and zeros elsewhere. Here, $q_{x, x+e_{ik}-e_{i0}}$ is the transition rate for agent $i$ to switch from no arm preference to preferring arm $a_k$, and $q_{x, x+e_{ik}-e_{ij}}$ is the transition rate for agent $i$ to switch from preferring arm $a_j$ to preferring arm $a_k$.
Furthermore, the Markov chain $\bm{X}^N (\cdot)$ admits the following sample path construction. For all $i\in[N]$ and for all $k=0,\dots,K$, we have
\begin{align*}
&\bm{X}_{i,k}^N (t)  = \bm{X}_{i,k}^N (0) \\
& \quad +
\int\limits_{[0,K)\times[0,t]} \pth{\indc{\bm{X}_{i,k}^N (s^-) =0} \indc{y \in C_{i,k}^{N, +} (s^-)} - \indc{\bm{X}_{i,k}^N (s^-) =1} \indc{y \in C_{i,k}^{N, -} (s^-)}} \mathcal{N}_i(dy,ds),
\end{align*}
where $\bm{X}_{i,k}^N (0)$ is the initial condition of $\bm{X}_{i,k}^N (t)$, $\mathcal{N}_i$ is a two-dimensional Poisson process of rate $\lambda$ over the set $[0,K)\times[0,\infty)$, and
\begin{align*}
C_{i,0}^{N, +}(t) &:= \emptyset, \\
C_{i,k}^{N, +} (t) &:= \left[k-1,\,\, k-1+ \frac{\mu p_k}{K}\indc{\bm{X}_{i,0}^N(t)=1} + \left(1-\mu\indc{\bm{X}_{i,0}^N(t)=1}\right) \frac{1}{|V_i^N|} \sum\limits_{\ell\in V_i^N} p_k \indc{\bm{X}_{\ell,k}^N (t) =1} \right), \\
C_{i,0}^{N,-}(t) &:= \bigcup_{j=1}^K \left[j-1,\,\, j-1+ \frac{\mu p_j}{K} + \left(1-\mu\right) \frac{1}{|V_i^N|} \sum\limits_{\ell\in V_i^N} p_j \indc{\bm{X}_{\ell, j}^N (t) =1} \right), \\
C_{i,k}^{N, -} (t) &:= \bigcup_{\overset{j=1}{j\neq k}}^K \left[j-1,\,\, j-1+  \frac{1}{|V_i^N|} \sum\limits_{\ell\in V_i^N} p_j \indc{\bm{X}_{\ell, j}^N (t) =1} \right).
\end{align*}
It turns out that this sample path construction significantly simplifies the statements and proofs of the results on the transient system behaviors.

\section{Main Results} \label{sec: main results}

In this section, we present the main results regarding {\em learnability} -- the probability that the entire social group learns the best option -- (subsection \ref{sec:learnability}), and regarding the transient of opinion evolution before reaching a consensus (subsection \ref{sec:transient}). Furthermore, in subsection \ref{sec:learnabilityVsTransient} we highlight the differences and similarities between the assumptions for the two types of results.

\subsection{Learnability}\label{sec:learnability}
For $k=0,\dots,K$, let
\begin{equation*}
  Z^N_k(t) := \sum\limits_{i=1}^N \bm{X}_{i,k}^N (t)
\end{equation*}
be the number of agents that prefer arm $a_k$ at time $t$. We define the success event $E^N$ as:
\begin{align}
\label{event: success}
\nonumber
E^N &\triangleq \sth{\text{every agent eventually learns the best option} } \\
&= \sth{\lim_{t\to\infty} Z^N_1(t) ~ = ~ N}.
\end{align}
Given the transition rates of the Markov chain $\bm{X}^N (\cdot)$, we have that, for $k=0, \cdots, K$, the process $Z^N_k(\cdot)$ jumps upwards with rate
\begin{align*}
&\sum_{i=1}^N \indc{\bm{X}_{i,k}^N(t) = 0}  \lambda p_k \left(
\frac{\mu}{K} \indc{\bm{X}_{i,0}^N(t) = 1} 
+ \left(1-\mu \indc{\bm{X}_{i,0}^N(t) = 1}\right) \frac{1}{D_i^N} \sum_{j \in V_i^N} \indc{\bm{X}_{j,k}^N(t) = 1}\right),
\end{align*}
%
and downwards with rate
\begin{align*}
&\sum_{i=1}^N \indc{\bm{X}_{i,k}^N(t) = 1}  \lambda  \left( \frac{\mu}{K} \indc{\bm{X}_{i,0}^N(t) = 1} 
+ \left(1-\mu \indc{\bm{X}_{i,0}^N(t) = 1}\right) \frac{1}{D_i^N} \sum_{j \in V_i^N} \sum_{k^{\prime}: 1\leq k^{\prime} \leq K, \& k^{\prime} \not=k}\indc{\bm{X}_{j,k^{\prime}}^N(t) = 1} p_{k^{\prime}} \right).
\end{align*}
It is easy to see that $Z^N_k(t)=N$ is an absorbing state of the Markov chain $\bm{X}^N(\cdot)$, for all $k=1,\dots,K$.

For ease of exposition, we first consider the simplified scenario wherein the memory states of the individuals are initialized deterministically, and there exists a positive fraction of agents with the correct opinion -- initially preferring the arm $a_1$ (i.e., the best arm). It turns out that the probability that eventually every agent pulls only arm $a_1$ increases exponentially with $N$. This is formalized in the following result.

\begin{theorem}
\label{thm: eventual learn}
Suppose that the graph $G_N$ is connected, and suppose that $Z_1^N(0)\geq c_0 N$ for some $c_0\in(0,1]$.
\begin{itemize}
	\item If $G_N$ is regular, then, for any $\mu \in [0,1]$, we have
\begin{align*}
&\mathbb{P}\big(E^N\big)\geq 1- \pth{\frac{p_1}{p_2}}^{-c_0N}.
\end{align*}

\item If $G_N$ is doubly-stochastic, i.e., if
\begin{align*}
\sum_{j \in V_i^N} \frac{1}{|V_j^N|} = 1, \qquad \forall\, i\in[N],
\end{align*}
then, for any $\mu\in[0,1-p_2/p_1]$, we have
\begin{align*}
&\mathbb{P}\big( E^N \big) \geq 1- \pth{\frac{(1-\mu + \frac{\mu}{K})p_1}{p_2}}^{-c_0N}.
\end{align*}
\end{itemize}
\end{theorem}

The proof consists in coupling the process $Z^N_1(\cdot)$ with a smaller but simpler biased random walk, and using classical results on biased random walk  to lower bound the probability of hitting $N$. The proof is deferred to Appendix \ref{app: limit}.

\begin{remark}
Intuitively, in a doubly-stochastic graph the ``outflow'' of agents equals the ``inflow'' of agents. \footnote{In fact, this assumption is one of the standard assumptions on reaching average consensus. } In particular, every regular graph is also doubly-stochastic.
\end{remark}

Note that Theorem \ref{thm: eventual learn} says that the probability of ``every agent eventually learns the best option'' grows to 1 exponentially fast as the group size $N$ increases. In order to maximize the lower bound in Theorem \ref{thm: eventual learn}, we can simply choose $\mu=0$ -- removing the uniform sampling. However, as we will show next, it is crucial to have $\mu >0$ when $Z_1^N(0)=0$.

\paragraph{\bf Generalizations}

In Theorem \ref{thm: eventual learn} we restrict our attention to initial conditions where a positive proportion of agents prefer arm $a_1$. This result can be easily generalized to include initial conditions where $Z_1^N(0)=0$ and $Z_0^N(0)\geq c_0 N$, for some $c_0\in(0,1]$. The key to this generalization is the following lemma.

\begin{lemma}
\label{lm: initial wealth}
Suppose that $Z_0^N(0)\geq c_0 N$ for some $c_0\in(0,1]$, and that $\mu>0$. Then,  for any $C \in (0, 1)$, we have
\[
\mathbb{P}\left( Z_1^N\pth{\frac{1}{\lambda}} \geq (1-C)\frac{\mu c_0p_1}{e K}N \right) \geq 1- \exp\pth{-2\pth{\frac{C\mu c_0p_1}{e K}}^2N}.
\]
\end{lemma}
The intuition behind Lemma \ref{lm: initial wealth} is that when $t$ is sufficiently small, successful memory state updates of the agents that initially have no preference over the $K$ arms mainly rely on uniform sampling rather than peer recommendation. Thus, successful memory updates of those agents are likely to be {\em independent} of each other, and have some nice concentration properties. We justify this intuition in Appendix \ref{app: initial wealth}.\\

Using Lemma \ref{lm: initial wealth}, we can couple our original random walk and a standard random walk at time $\frac{1}{\lambda}$, and obtain an exponential learnability result analogous to Theorem \ref{thm: eventual learn}.\\

Furthermore, we can generalize Theorem \ref{thm: eventual learn} and Lemma \ref{lm: initial wealth} for the case where the initial conditions $\{\bm{X}^N_i(0)\}_{i\in[N]}$ are i.i.d. with respect to a fixed
distribution $q=(q_0,q_1,\dots,q_K)$.

\begin{lemma}
Suppose that $\{\bm{X}^N_i(0)\}_{i\in[N]}$ are i.i.d. with respect to a fixed
distribution $q=(q_0,q_1,\dots,q_K)$, with $q_0+q_1>0$. Then, for every $C\in(0,1)$, we have
\begin{align*}
&\mathbb{P}\Big( Z^N_0(0) + Z^N_1(0) \geq (1-C)(q_0+q_1)N \Big)
~ \geq ~ 1- \exp\Big( - 2C^2(q_0+q_1)^2 N \Big).
\end{align*}
\end{lemma}
The proof is a simple application of Hoeffding's inequality, and is omitted.

\subsection{Transient System Behaviors}\label{sec:transient}
In addition to {\em learnability}, it is also important to characterize the transient behavior of our learning dynamics, 
i.e., at a given time $t$, what fraction of agents are preferring the best arm, the second best arm, etc. This is because in applications such as biology, chemistry, and networking, knowing the trajectories of the dynamics of interests usually provide fundamental insights on those systems. This subsection is devoted to characterizing this transient system behaviors.\\


First, we will establish a local approximation result. For $i\in[N]$, and for $k=0,\dots,K$, we define the coupled processes
\begin{align}
\bm{X}_{i,k} (t) & = \bm{X}_{i,k}^N (0) +
\int\limits_{[0,K)\times[0,t]} \pth{\indc{\bm{X}_{i,k} (s^-) =0} \indc{y \in C_{i,k}^+(s^-)} -  \indc{\bm{X}_{i,k} (s^-) =1} \indc{ y \in C_{i,k}^-(s^-)} } \mathcal{N}_i(dy,ds). \label{eq:limitProcess}
\end{align}
where
\begin{align*}
C_{i,0}^+(t) &:= \emptyset, \\
C_{i,k}^+(t) &:= \left[k-1,\,\, k-1+ \frac{\mu p_k}{K}\indc{\bm{X}_{i,0}(t)=1} + p_k \big(1-\mu\indc{\bm{X}_{i,0}(t)=1}\big) \mathbb{P}\left( \bm{X}_{i, k} (t) =1 \right) \right), \qquad \forall\, k\geq 1, \\
C_{i,0}^-(t) &:= \bigcup_{j=1}^K \left[j-1,\,\, j-1+ \frac{\mu p_j}{K} + (1-\mu) p_j \mathbb{P}\left( \bm{X}_{i, j} (t) =1 \right) \right),\\
C_{i,k}^-(t) &:= \bigcup_{\overset{j=1}{j\neq k}}^K \Big[j-1,\,\, j-1+ p_j \mathbb{P}\left( \bm{X}_{i, j} (t) =1 \right) \Big), \qquad \forall\, k\geq 1.
\end{align*}
Note that the processes $\bm{X}^N(\cdot)$ and $\bm{X}(\cdot)$ are coupled through the initial conditions $\bm{X} (0) = \bm{X}^N (0)$
and through the underlying Poisson processes $\mathcal{N}_i$ (they are all the same).
However, unlike $C_{i,k}^{N,-}(t)$ and $C_{i,k}^{N,+}(t)$, the sets $C_{i,k}^{-}(t)$ and $C_{i,k}^{+}(t)$ only depend on the local state. This corresponds to a setting where an agent, instead of asking a neighbor for its state, it draws a new state (independent from the states of its neighbors), according to the distribution $\big(\mathbb{P}\left( \bm{X}_{i, 0} (t) =1 \right),\dots,\mathbb{P}\left( \bm{X}_{i, K} (t) =1 \right)\big)$.

We now show that, if the minimum degree of the graph $G_N$ is large enough, and if the initial conditions are i.i.d., then the coordinate processes $\{\bm{X}^N_i(\cdot)\}_{i\in[N]}$ can be approximated by the i.i.d. coordinate processes $\{\bm{X}_i(\cdot)\}_{i\in[N]}$ defined above. This is formalized in the following result.

\begin{theorem}\label{thm:localConvergence}
Fix some time $T>0$. Suppose that the initial conditions $\{\bm{X}_{i}^N (0)\}_{i\in[N]}$ are i.i.d. with respect to a fixed distribution $q=(q_0,\dots,q_K)$, such that
$\mathbb{P}(\bm{X}_{i,k}^N (0)=1)=q_k$,
for all $k=0,\dots,K$, and for all $i\in[N]$. Then
\begin{align*}
&\max\limits_{i\in[N]} \expect{\sup_{0\leq t \leq T} \sum_{k=0}^K \left(\bm{X}_{i, k}^N(t) - \bm{X}_{i, k}(t)\right)^2 } 
\leq \frac{16(4+\lambda)\lambda KT(K+1)}{\sqrt{D_{\min}^N}} \exp\Big( 48(4+\lambda)\lambda(K+1)T \Big).
\end{align*}
\end{theorem}
The proof is deferred to Appendix \ref{app:localConvergence}.\\

\begin{remark}
Theorem \ref{thm:localConvergence} states that, for $D^N_{min}$ large enough, the coordinate processes $\{\bm{X}^N_i(\cdot)\}_{i\in[N]}$ can be approximated by the i.i.d. coordinate processes $\{\bm{X}_i(\cdot)\}_{i\in[N]}$. In particular, this implies that the original coordinate processes are asymptotically independent.
\end{remark}

Using this local approximation result and i.i.d. processes, we establish a convergence result on $Z_k^N(t)$.
For any given $N$, and for $k=0,\dots,K$, let
\begin{equation*}
  Y^N_k(t) := \frac{1}{N} Z_k^N(t) = \frac{1}{N} \sum\limits_{i=1}^N \bm{X}_{i,k}^N (t)
\end{equation*}
be the fraction of agents that prefer arm $a_k$ at time $t$. \footnote{Recall that we say an agent prefers arm $a_k$ at time $t$ is its memory state is $M=k$ at time $t$. } We will show that over any fixed finite time horizon $[0, T]$, $Y^N(\cdot)=\big(Y_0^N(\cdot),\dots,Y_K^N(\cdot)\big)$ converges uniformly to the unique\footnote{The uniqueness follows from the fact that the drift is a simple polynomial.} solution $y(\cdot)$ of the following ODEs.
\begin{align}
\label{limit 0}
\dot{y_0}(t) &= - y_0(t) \lambda \frac{\mu}{K}\sum_{j=1}^K p_j - y_0(t) \lambda \sum_{j=1}^K (1-\mu) p_j y_j(t),\\
\nonumber
\dot{y_k}(t) &=  y_0(t) \lambda  \frac{\mu}{K} p_k + y_k(t) \lambda \pth{(1-\mu)p_k y_0(t)+ \sum_{j=1}^K (p_k-p_{j})y_{j}(t)},\\
& \qquad \qquad \qquad \qquad \qquad \qquad \qquad \qquad  \qquad  \forall\,  k\ge 1.
\label{limit 1}
\end{align}
This convergence result is formalized in the following theorem.

\begin{theorem} \label{thm:transient}
Fix some time $T>0$. Suppose that the initial conditions $\{\bm{X}_{i}^N (0)\}_{i\in[N]}$ are i.i.d. with respect to a fixed distribution $q=(q_0,\dots,q_K)$, such that $\mathbb{P}(\bm{X}_{i,k}^N (0)=1)=q_k$, for all $k=0,\dots,K$, and for all $i\in[N]$. Furthermore, suppose that
\[ \lim_{N\to\infty} D_{\min}^N = \infty. \]
Then,
\[ \lim\limits_{N\to\infty} \expect{\sup_{0\leq t \leq T} \sum_{k=0}^K \left( Y^N_k(t) - y_k(t)\right)^2 } =0, \]
where $y(t)$ is the solution to the ODEs defined by equations \eqref{limit 0} and \eqref{limit 1} with initial condition $y(0)=q$.
\end{theorem}
The proof is deferred to Appendix \ref{app: detailed proof of mean-field}.\\

This theorem implies that, for $N$ large enough, the process $Y^N(t)$ closely tracks a deterministic and smooth trajectory. We will illustrate this via a simulation result at the end of this subsection. 
Such approximations are desired because the analysis of an ODEs system is relatively easier than that of the original stochastic system. Indeed, in a great variety of fields, such as biology, epidemic theory, physics, and chemistry \cite{kurtz1970solutions}, differential equations are used directly to model the {\em macroscopic} level system dynamics that are {\em arguably} caused by the {\em microscopic} level agents interactions in the system.

Note that the above ODE system is similar to the antisymmetric Lotka-Volterra equation \cite{knebel2015evolutionary} with $\mu=0$. The Lotka-Volterra equation is a typical replicator dynamics, where if $y_k(0)=0$ for some $k$, it remains to be zero throughout the entire process. In contrast, even when $y_1(0)=0$, the solution of our ODE system converges exponentially fast to $\pth{0,1,0,\dots,0}$ in time as long as $\mu>0$ and $y_0(0)>0$. 
%

\begin{theorem}\label{thm:expConvergence}
  Let $y$ be the solution of the ODEs in \eqref{limit 0} and \eqref{limit 1}. We have the following.
  \begin{itemize}
    \item Suppose $y_1(0)>0$. Then, we have
    \[ y_1(t)\geq 1- \frac{1}{\frac{y_1(0)}{1-y_1(0)} \exp \pth{R t } +1}, \qquad \forall \,\, t\geq 0, \]
    where
  \[ R=\lambda \times \min\left\{\left(1-\mu+\frac{\mu}{K}\right)p_1, ~~ p_1-p_2 \right\}. \]
    \item Suppose $y_1(0)=0$ and $y_0(0)>0$. Then, for every $c\in(0,1)$, we have
    \[ y_1(t)\geq 1- \frac{1}{\frac{(1-c)y_0(0)}{K-(1-c)y_0(0)} \exp \pth{R (t-t_c) } +1}, \qquad \forall \,\, t\geq \bar{t}_c, \]
    where
    \[ \bar{t}_c \triangleq \frac{\log \frac{1}{c}}{\lambda \frac{\mu}{K}\sum_{i=1}^K p_i}. \]
  \end{itemize}
\end{theorem}

The proof of the second part involves showing that at time $\bar{t}_c$, we have $y_1(\bar{t}_c)\geq \frac{(1-c)y_0(0)}{K}$, and then obtaining a suitable lower bound that holds from that point going forward. The proof of the first part is a special case with $\bar{t}_c=0$. The proof is in Appendix \ref{app:proofExpConvergence}.\\

An immediate corollary of the previous theorem is the asymptotic convergence of the trajectories to the desired state.

\begin{corollary}\label{cor:convergence}
Let $y$ be the solution of the ODEs in \eqref{limit 0} and \eqref{limit 1} with $y_0(0)+y_1(0)>0$. We have that
 \[ \lim\limits_{t\to\infty} y(t) = \pth{0,1,0,\dots,0}. \]
\end{corollary}

The results in Theorems \ref{thm:transient}, \ref{thm:expConvergence}, and Corollary \ref{cor:convergence} are illustrated in Figure \ref{example}, where some typical sample paths are drawn.
\begin{figure}
\centering
\includegraphics[width=15cm]{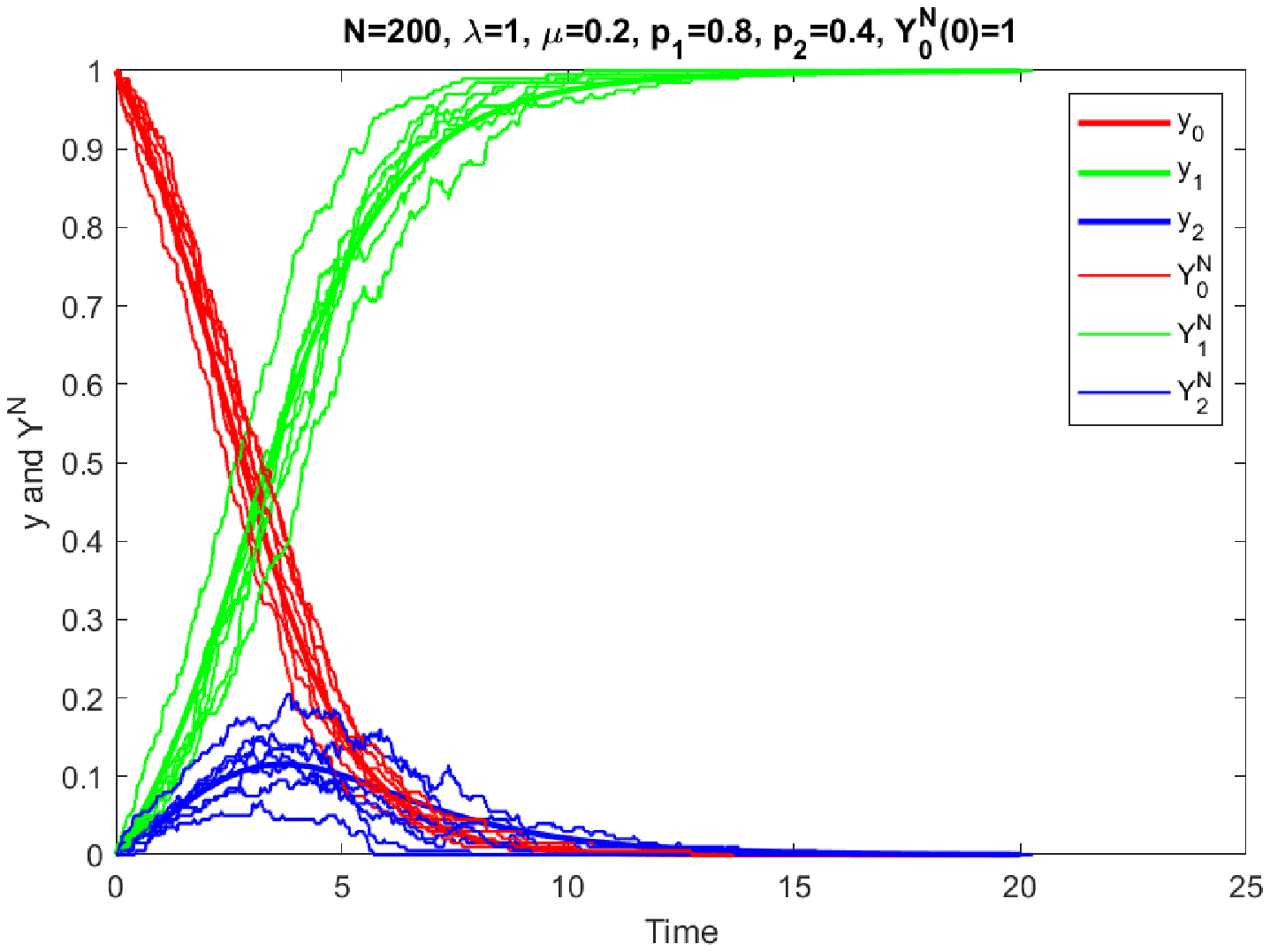}
\caption{} 
\label{example}
\end{figure}
Here, we choose  $N=200$, $K=2$, $\lambda=1$, $\mu=0.2$, $p_1=0.8$, $p_2=0.4$, and initial conditions $Y^N=(1, 0, \cdots, 0)$ (i.e., every agent starts with no preference over the arms). 
In Figure \ref{example}, each of the component in $Y^N$ goes to their corresponding equilibrium states exponentially fast. In particular, these typical sample paths have two slightly different behaviors stages: At the first stage, $Y^N_1(t)$ and $Y^N_2(t)$ both increase up to the point where $Y^N_1(t)+Y^N_2(t) \approx 1$ -- noting that $Y^N_2(t)$ grows much slower than $Y^N_1(t)$. At the second stage, until entering their equilibrium states, $Y^N_1(t)$ is increasing and $Y^N_2(t)$ is decreasing. More importantly, $Y_0^N, Y_1^N,$ and $Y_2^N$ track their corresponding deterministic and smooth trajectories, which converge to the desired $(0, 1, \cdots, 0)$ in time.

\subsection{Learnability vs transient and interchange of limits} \label{sec:learnabilityVsTransient}

Recall that, for the learnability results in subsection \ref{sec:learnability}, we need the graph $G_N$ to be connected and either doubly-stochastic or regular, and we need the initial conditions be such that $Z_1^N(0)+Z_1^N(0)>cN$ (holds either deterministically or with high probability) for some positive constant $c>0$.
Under these conditions, we have
\[ \lim\limits_{N\to\infty} \lim\limits_{t\to\infty} Y^N_1(t) = \lim\limits_{N\to\infty} \mathds{1}_{E_N} = 1. \]
Moreover, for the transient approximation result and the convergence to the best state presented in subsection \ref{sec:transient}, we need the graph $G_N$ to have diverging minimum degree, and i.i.d. initial conditions with $q_0+q_1>0$. Under these conditions, we have
\[ \lim\limits_{t\to\infty} \lim\limits_{N\to\infty} Y^N_1(t) = \lim\limits_{t\to\infty} y_1(t) = 1. \]
When we have all the assumptions at the same time, we can interchange the order of these limits. This is depicted in the commutative diagram of Figure \ref{fig:limits_diagram}.\\

\begin{figure}[ht!]
  \centering
  \begin{tikzpicture}[scale=1]
    \draw (0,0) node {$\mathds{1}_{E_N}$};
    \draw [<-,thick] (0,0.4) -- (0,2.6);
    \draw (1.5,0.3) node {Thm. \ref{thm: eventual learn}};
    \draw (1.5,-0.2) node {$N\to\infty$};
    \draw [->,thick] (0.4,0) -- (2.6,0);
    \draw (-0.8,1.7) node {Eq. \eqref{event: success}};
    \draw (-0.8,1.3) node {$t\to\infty$};
    \draw (0,3) node {$Y^N_1(t)$};
    \draw [->,thick] (0.4,3) -- (2.6,3);
    \draw (3.75,1.7) node {Cor. \ref{cor:convergence}};
    \draw (3.75,1.3) node {$t\to\infty$};
    \draw (3,3) node {$y_1(t)$};
    \draw [->,thick] (3,2.6) -- (3,0.4);
    \draw (1.5,3.3) node {Thm. \ref{thm:transient}};
    \draw (1.5,2.8) node {$N\to\infty$};
    \draw (3,0) node {$1$};
  \end{tikzpicture}
  \caption{Interchange of limits.}
  \label{fig:limits_diagram}
\end{figure}

On the other hand, the differences in the assumptions are not an artifact of our analysis, but they stem from fundamental differences in the dynamics. This is shown in the following examples.

\begin{example}
  Let $G_N$ be the circular graph with $N$ agents. Suppose that the initial condition is such that there are $N/2$ agents that prefer arm $a_1$ and $N/2$ agents that prefer arm $a_2$. Then, Theorem \ref{thm: eventual learn} states that the probability that all agents eventually learn the best arm converges to $1$ exponentially fast in $N$. However, if the initial condition is such that all agents prefer arm $a_1$ are all next to each other, then there are only 4 agents that have a neighbor with a different preferred arm than theirs (the two pair of agents in the boundary of preferred arms). It can be checked that $Z_1^N(\cdot)$ is a birth-death process with birth rate $2\lambda p_1$ and death rate $2\lambda p_2$. As a result, the approximation result given by Theorem \ref{thm:transient} does not hold. Not only the transient is now linear instead of exponential, but the speed of convergence to the best arm is $N$ times slower.
\end{example}

\begin{example}
  Let $G_N$ be the graph consisting of $N/\log\log(N)$ connected components consisting of complete graphs with $\log\log(N)$ agents. If the initial conditions are i.i.d. with $q_0+q_1>0$, then Theorem \ref{thm:transient} holds and we have the usual exponential convergence to the desired state. However, if $q_i>0$ for some $i>1$, it can be checked that all agents in at least one connected component will prefer arm $a_i$, with high probability. Since those agents will prefer arm $a_i$ for all time, not all agents will eventually learn the best arm, with high probability.
\end{example}

\section{Concluding Remarks}
\label{subsec: comparison}
We studied the collaborative multi-armed bandit problems in social groups wherein each agent suffers finite memory constraint \cite{1054427}.
In contrast to isolated agents \cite{1054427} for whom learning the best option is impossible, we showed that with the aid of social persuasion even if the communication graphs can be highly sparse, the probability of collaboratively learning the best option goes to 1 exponentially fast in $N$.
We also characterized the transient system behaviors. In particular, we showed that if the minimum degree of the graph diverges as $N \diverge $, sample paths describing the opinion evolutions of the individuals are asymptotically independent.  Additionally, over an arbitrary but given finite time horizon, the proportions of population with different opinions converge to the unique solution of a system of ODEs that are invariant to the structures of the communication graphs. In the solution of the obtained ODEs, the proportion of the population holding the correct opinion converges to $1$ exponentially fast in time. The key challenge in the analysis of general graphs is that due to the lack of exchangeability, one needs to keep track of the opinion evolution of each individual. This complicates the state description of the system. %



\section*{Acknowledgements}
We would like to thank John N.\ Tsitsiklis and Debankur Mukherjee for valuable discussions and comments.

\bibliographystyle{acm}
\bibliography{alpha}

\appendix

\section{Proof of Theorem \ref{thm: eventual learn}}
\label{app: limit}

The main idea in proving Theorem \ref{thm: eventual learn} is to couple the process $Z^N_1(\cdot)$  with a standard biased random walk whose probability of hitting $N$ (i.e., success probability) is well understood.

Despite the fact that the $N\times (K+1)$-dimensional Markov chain $\bm{X}^N(\cdot)$ is hard to directly analyze, the evolution of the number of agents with the correct opinion (i.e., preferring the best arm $a_1$)
$Z^N_1(\cdot)$ has the following nice property.

\begin{lemma}
\label{lm: random walker}
Suppose that $G_N$ is connected and doubly-stochastic. If there is a jump in $Z^N_1(\cdot)$ at time $t$, the probability of moving upwards is lower bounded as
\begin{align*}
&\mathbb{P} \Big( Z^N_1(t) = Z^N_1(t^-) +1 \,\,\Big|\,\, Z^N_1(t) \not= Z^N_1(t^-) \Big) 
\geq \frac{p_1(1-\mu) + \frac{\mu}{K} p_1}{p_1(1-\mu) + \frac{\mu}{K} p_1 + p_2}.
\end{align*}
In addition, if $G_N$ is regular, then
\begin{align*}
\mathbb{P} \Big( Z^N_1(t) = Z^N_1(t^-) +1 \,\,\Big|\,\, Z^N_1(t) \not= Z^N_1(t^-) \Big)  \geq \frac{p_1}{p_1+p_2}.
\end{align*}
\end{lemma}
We prove this lemma in Appendix \ref{app: random walker}.\\

Let
\[ \{T_l\}_{l\geq 1} \triangleq \big\{t\geq 0: Z^N_1(t)\neq Z^N_1(t^-) \big\} \]
be the set of random times when there is a jump in the process $Z^N_1(\cdot)$. Without loss of generality, we assume that $T_1\leq T_2 \leq \cdots$. Note that the set $\{T_l\}_{l\geq 1}$ might be finite since $\bm{X}^N(\cdot)$ might be absorbed after a finite number of jumps. For notational convenience, denote the random walk
\[
W(l)~ \triangleq ~ Z^N_1(T_l).
\]
Clearly, $W(\cdot)$ is a random walk that represents the evolution of the number of agents with the correct opinion (i.e., preferring the best arm $a_1$). 
Thus, $E^N$ can be rewritten as
$$E^N = \sth{\lim_{l\to\infty} W(l) = N}.$$
It follows from Lemma \ref{lm: random walker} that
\begin{align*}
&\mathbb{P}\left(W(l+1) = W(l)+1 \,\,\left|\,\, \text{the $(l+1)$-st jump of $Z^N_1(\cdot)$ exists} \right.\right) \\
& 
\geq
\begin{cases}
\frac{p_1}{p_1+p_2}, & ~~ \text{if $G_N$ is regular}, \\
\frac{p_1(1-\mu) + \frac{\mu}{K} p_1}{p_1(1-\mu) + \frac{\mu}{K} p_1 + p_2}, & ~~ \text{if $G_N$ is doubly-stochastic}.
\end{cases}
\end{align*}
When $G_N$ is regular, the random walk $W(\cdot)$ is biased towards moving upwards.
When $G_N$ is only doubly-stochastic, 
we need to have $\mu\leq 1-p_2/p_1$ in order to guarantee that $W(\cdot)$ is also biased towards moving upwards.

Consider the standard random walk $\pth{\widehat{W}(l): \, l\in \integers_+}$ such that if $\widehat{W}(l) =0$ or $\widehat{W}(l) =N$, then $\widehat{W}(l+1) = \widehat{W}(l)$; otherwise,
\begin{align}
\label{aux: random walk}
\widehat{W}(l+1) =
\begin{cases}
\widehat{W}(l)+1 ~ & \text{with probability  } p^*; \\
\widehat{W}(l)-1  ~ & \text{with probability  } 1-p^*,
\end{cases}
\end{align}
where
\begin{align*}
p^* =
\begin{cases}
\frac{p_1}{ p_1+p_2}, & \text{if $G_N$ is regular}; \\
\frac{p_1(1-\mu) + \frac{\mu}{K} p_1}{p_1(1-\mu) + \frac{\mu}{K} p_1 + p_2}, & \text{if $G_N$ is doubly-stochastic}.
\end{cases}
\end{align*}
From \cite[Chapter 4.2]{hajek2015random}, we know that for any $z_0\in \integers_+$,
%
\[
\mathbb{P}\left(\left. \lim_{l\diverge} \widehat{W}(l) =N \,\,\right|\,\, \widehat{W}(0)=z_0 \right) 
   \geq 1- \pth{\frac{1-p^*}{p^*}}^{- z_0}.
\]

Intuitively, the original random walk $W(l)$ has a higher tendency to move one step up (if possible) than that of the standard random walk \eqref{aux: random walk}. Thus, starting at the same position, the original random walk has a higher chance to be absorbed at position $N$ than that of the standard random walk. Thus,
\begin{align*}
\mathbb{P}\left(\left. \lim_{l\diverge} W(l) =N \,\,\right|\,\, W(0)=z_0 \right) &\geq \mathbb{P}\left(\left. \lim_{l\diverge} \widehat{W}(l) =N \,\,\right|\,\, \widehat{W}(0)=z_0 \right)  \\
&\geq  1- \pth{\frac{1-p^*}{p^*}}^{- z_0}. 
\end{align*}
We conclude Theorem \ref{thm: eventual learn} by choosing $z_0 \ge c_0N$. 

\subsection{Proof of Lemma \ref{lm: random walker}}
\label{app: random walker}

The proof of Lemma \ref{lm: random walker} uses the following proposition.
\begin{proposition}
\label{prop: symmetry}
Suppose that $G_N$ is doubly-stochastic. Then for any subset $\calS$ of $[N]$, it holds that
\begin{align*}
\sum_{i\in \calS} \frac{1}{|V_i^N|} \sum_{j\in V_i^N} \indc{j\notin \calS} = \sum_{i\notin \calS} \frac{1}{|V_i^N|} \sum_{j\in V_i^N} \indc{j\in \calS}.
\end{align*}
\end{proposition}
\begin{remark}
Proposition \ref{prop: symmetry} says that if $G_N$ is doubly-stochastic, for any set $\calS\subseteq [N]$ in graph $G_N$, the ``flow'' out of set $\calS$ equals the ``flow'' into set $\calS$. 
\end{remark}
%
\begin{proof}[Proof of Proposition \ref{prop: symmetry}]
When $\calS=\emptyset$ or $\calS=[N]$, it is easy to see that
\begin{align*}
\sum_{i\in \calS} \frac{1}{|V_i^N|} \sum_{j\in V_i^N} \indc{j\notin \calS}  = 0 = \sum_{i\notin \calS} \frac{1}{|V_i^N|} \sum_{j\in V_i^N} \indc{j\in \calS}.
\end{align*}
For the more general subset of $\calS$ we have
\begin{align*}
\sum_{i\in \calS} \frac{1}{|V_i^N|} \sum_{j\in V_i^N} \indc{j\notin \calS} 
& =  \sum_{i=1}^N  \frac{1}{|V_i^N|} \sum_{j\in V_i^N}  \indc{j\notin \calS} -  \sum_{i\notin \calS} \frac{1}{|V_i^N|} \sum_{j\in V_i^N} \pth{1-\indc{j\in \calS}}\\
& = \sum_{i\notin \calS} \frac{1}{|V_i^N|} \sum_{j\in V_i^N} \indc{j\in \calS} + \sum_{i=1}^N  \frac{1}{|V_i^N|} \sum_{j\in V_i^N} \indc{j\notin \calS} 
-\sum_{i\notin \calS} \frac{1}{|V_i^N|} \sum_{j\in V_i^N} 1.
\end{align*}
To finish the proof, it remains to show that
\begin{align}
\label{eq: bound1}
\sum_{i=1}^N \frac{1}{|V_i^N|} \sum_{j\in V_i^N} \indc{j\notin \calS}=\sum_{i\notin \calS} \frac{1}{|V_i^N|} \sum_{j\in V_i^N} 1.
\end{align}
The RHS of \eqref{eq: bound1} can be written as
\begin{align}
\label{eq: rhs}
\nonumber
\sum_{i\notin \calS} \frac{1}{|V_i^N|} \sum_{j\in V_i^N} 1 & = \sum_{i\notin \calS} 1
\overset{(a)}{=} \sum_{i\notin \calS}  \sum_{j \in V_i^N} \frac{1}{|V_j^N|} \\ 
 &= \sum_{i=1}^N \sum_{j=1}^N \indc{i\notin \calS} \indc{(i,j)\in \calE} \frac{1}{|V_j^N|}.
\end{align}
where (a) follows from the fact that $G_N$ is doubly-stochastic.
The LHS of \eqref{eq: bound1} can be written as
\begin{align*}
\nonumber
\sum_{i=1}^N \frac{1}{|V_i^N|} \sum_{j\in V_i^N} \indc{j\notin \calS} 
& = \sum_{i=1}^N \sum_{j=1}^N \indc{(i,j)\in \calE^N} \frac{1}{|V_i^N|} \indc{j\notin \calS}\\
 &= \sum_{j=1}^N \sum_{i=1}^N \indc{(j,i)\in \calE^N} \frac{1}{|V_j^N|} \indc{i\notin \calS},
\end{align*}
proving \eqref{eq: bound1}.  

\end{proof}

Now we are ready to prove Lemma \ref{lm: random walker}.
\begin{proof}[Proof of Lemma \ref{lm: random walker}]
At any time $t$, if there is a jump in $Z_1^N(t)$, it must be true that
\begin{align}
\label{eq: condition}
Z_1^N(t-) \not=N, ~ \text{and} ~
Z_1^N(t-) +  Z_0^N(t-) ~ > ~ 0.
\end{align}
Note that it is possible that $ Z_0^N(t-)=N$.
Let $\calS(t-)$, $\calA(t-)$ and $\calB(t-)$ be a partition of set $[N]$ such that:
\begin{itemize}
\item Let $\calS(t-)$ denote the set of agents that currently prefer arm $a_1$, i.e.,
\begin{align*}
\calS(t-) : = \sth{i\in [N]: ~ \bm{X}_{i,1}^N(t-) = 1}.
\end{align*}
\item Let $\calA(t-)$ denote the set of agents that currently have no preference over the $K$ arms, i.e.,
\begin{align*}
\calA(t-) : = \sth{i\in [N]: ~ \bm{X}_{i,0}^N(t-) = 1}.
\end{align*}
\item Let $\calB(t-) := [N] -\calS(t-) - \calA(t-)$, i.e.,
\begin{align*}
\calB(t-) : = \sth{i\in [N]: ~ \sum_{k=2}^{K} \bm{X}_{i,k}^N(t-) = 1}.
\end{align*}
\end{itemize}
Note that $\abth{\calS(t-)} = Z_1^N(t-)$, $\abth{\calA(t-)} = Z_0^N(t-)$, and $\abth{\calB(t-)} = \sum_{k=2}^K Z_k^N(t-).$

The upwards drift of $Z_1^N(t-)$
can be written as
\begin{align}
\label{eq: ccc}
\lambda p_1 \frac{\mu}{K} |\calA(t-)| + \pth{1-\mu} \lambda p_1\sum_{i\in \calA(t-)} \frac{1}{|V_i^N|} \sum_{j \in V_i^N} \indc{j\in \calS(t-)} 
+ \lambda p_1 \sum_{i\in \calB(t-)} \frac{1}{|V_i^N|} \sum_{j \in V_i^N} \indc{j\in \calS(t-)}.
\end{align}
Similarly, the downwards drift of $Z_1^N(t-)$ can be written as
\begin{align}
\label{eq: aaa}
\nonumber
\lambda \sum_{i\in \calS(t-)} \frac{1}{|V_i^N|} \sum_{j \in V_i^N} \sum_{k=2}^K \indc{\bm{X}_{j,k}^N(t-) = 1} p_{k}
& \leq \lambda p_2 \sum_{i\in \calS(t-)} \frac{1}{|V_i^N|} \sum_{j \in V_i^N} \sum_{k=2}^K \indc{\bm{X}_{j,k}^N(t-) = 1}\\
& =  \lambda p_2 \sum_{i\in \calS(t-)} \frac{1}{|V_i^N|} \sum_{j \in V_i^N} \indc{j\in \calB(t-)}.
\end{align}

We consider two cases: (1) there are no links between $\calS(t-)$ and $\calB(t-)$, and (2) there is a link between $\calS(t-)$ and $\calB(t-)$.

\paragraph{Case 1:} Suppose that there are no links between $\calS(t-)$ and $\calB(t-)$.
In this case, it holds that
\begin{align}
\label{case: no links BS}
\sum_{i\in \calS(t-)} \frac{1}{|V_i^N|} \sum_{j \in V_i^N} \indc{j\in \calB(t-)} =0.
\end{align}
By \eqref{eq: aaa}, the downwards drift of $Z_1^N(t-)$ is zero. The upwards drift of $Z_1^N(t-)$ can be written as
\begin{align*}
\lambda p_1 \frac{\mu}{K} |\calA(t-)| + \pth{1-\mu} \lambda p_1\sum_{i\in \calA(t-)} \frac{1}{|V_i^N|} \sum_{j \in V_i^N} \indc{j\in \calS(t-)}.
\end{align*}

Suppose that $\calA(t-)=\emptyset$. By \eqref{case: no links BS} and the fact that $G_N$ is connected, we know that either $\abth{\calS(t-)}=0$ or $\abth{\calS(t-)}=N$, contradicting \eqref{eq: condition}. Thus, $\calA(t-)\not=\emptyset$. So, the upwards drift of $Z_1^N(t-)$ is nonzero.  
Therefore, by \cite[Proposition 4.10]{hajek2015random}, we have
\begin{align*}
&\mathbb{P}\left(\left. Z_1^N(t) = Z_1^N(t-) +1 \,\,\right|\,\, Z_1^N(t)\not=Z_1^N(t-) \right)   = 1.
\end{align*}

\paragraph{Case 2:} Suppose that there exists a link between $\calS(t-)$ and $\calB(t-)$. Thus,
\begin{align}
\label{eq: case 2}
\sum_{i\in \calS(t-)} \frac{1}{|V_i^N|} \sum_{j \in V_i^N} \indc{j\in \calB(t-)} >0.
\end{align}
We consider two sub-cases:
\begin{itemize}
\item [Case 2-1: ]
\[
\sum_{i\in \calB(t-)} \frac{1}{|V_i^N|} \sum_{j \in V_i^N} \indc{j\in \calS(t-)} \geq \sum_{i\in \calS(t-)} \frac{1}{|V_i^N|} \sum_{j \in V_i^N} \indc{j\in \calB(t-)};
\]
\item [Case 2-2: ]
\[
\sum_{i\in \calB(t-)} \frac{1}{|V_i^N|} \sum_{j \in V_i^N} \indc{j\in \calS(t-)} < \sum_{i\in \calS(t-)} \frac{1}{|V_i^N|} \sum_{j \in V_i^N} \indc{j\in \calB(t-)}.
\]
\end{itemize}
Note that when $G_N$ is a regular graph, case 2-1 always holds.
To see this, let $G_N$ be $D$-regular, where $D\ge 1$.
We have
\begin{align*}
\sum_{i\in \calB(t-)} \frac{1}{|V_i^N|} \sum_{j \in V_i^N} \indc{j\in \calS(t-)} 
& = \frac{1}{D} \sum_{i=1}^N \sum_{j=1}^N \indc{(i,j)\in \calE^N}\indc{i\in \calB(t-)} \indc{j\in \calS(t-)}\\
& = \sum_{i\in \calS(t-)} \frac{1}{|V_i^N|} \sum_{j \in V_i^N} \indc{j\in \calB(t-)}.
\end{align*}

Henceforth, we consider general doubly-stochastic graphs.

\paragraph{Case 2-1:}
Suppose that
\begin{align}
\label{eq: case: more inflow}
\sum_{i\in \calB(t-)} \frac{1}{|V_i^N|} \sum_{j \in V_i^N} \indc{j\in \calS(t-)} \geq \sum_{i\in \calS(t-)} \frac{1}{|V_i^N|} \sum_{j \in V_i^N} \indc{j\in \calB(t-)}.
\end{align}
The downward drift in \eqref{eq: aaa} can be bounded as
\begin{align*}
 \lambda \sum_{i\in \calS(t-)} \frac{1}{|V_i^N|} \sum_{j \in V_i^N} \sum_{k=2}^K \indc{\bm{X}_{j,k}^N(t-) = 1} p_{k}
&\leq  \lambda p_2 \sum_{i\in \calB(t-)} \frac{1}{|V_i^N|} \sum_{j \in V_i^N} \indc{j\in \calS(t-)}.
\end{align*}
Thus, the total rate of moving away from state $Z_1^N(t-)$ is upper bounded as
\begin{align}
\label{eq: eee}
\lambda p_1 \frac{\mu}{K} |\calA(t-)| + \pth{1-\mu} \lambda p_1\sum_{i\in \calA(t-)} \frac{1}{|V_i^N|} \sum_{j \in V_i^N} \indc{j\in \calS(t-)} 
+ \lambda (p_1+p_2)\sum_{i\in \calB(t-)} \frac{1}{|V_i^N|} \sum_{j \in V_i^N} \indc{j\in \calS(t-)}.
\end{align}
Thus, we have
\begin{align*}
 \mathbb{P}\left(\left. Z_1^N(t) = Z_1^N(t-) +1 \,\,\right|\,\, Z_1^N(t)\not=Z_1^N(t-) \right) 
& \ge \frac{\text{Eq.} ~\eqref{eq: ccc} ~~ }{\text{Eq.} ~ \eqref{eq: eee} ~~ }
\ge \frac{p_1}{ p_1 + p_2}.
\end{align*}

\paragraph{Case 2-2:}
Suppose that
\begin{align}
\label{eq: case: less inflow}
\sum_{i\in \calB(t-)} \frac{1}{|V_i^N|} \sum_{j \in V_i^N} \indc{j\in \calS(t-)} < \sum_{i\in \calS(t-)} \frac{1}{|V_i^N|} \sum_{j \in V_i^N} \indc{j\in \calB(t-)}.
\end{align}
By Proposition \ref{prop: symmetry}, we have
\begin{align*}
& \sum_{i \in \calS(t-)} \frac{1}{|V_i^N|} \sum_{j\in V_i^N} \indc{j\in \calB(t-)} + \sum_{i \in \calS(t-)} \frac{1}{|V_i^N|} \sum_{j\in V_i^N} \indc{j\in \calA(t-)} \\
& =  \sum_{i \in \calB(t-)} \frac{1}{|V_i^N|} \sum_{j\in V_i^N} \indc{j\in \calS(t-)} + \sum_{i \in \calA(t-)} \frac{1}{|V_i^N|} \sum_{j\in V_i^N} \indc{j\in \calS(t-)}.
\end{align*}
Rearrange the terms, we get
\begin{align}
\label{eq: diff}
\nonumber
& \sum_{i \in \calS(t-)} \frac{1}{|V_i^N|} \sum_{j\in V_i^N} \indc{j\in \calB(t-)} - \sum_{i \in \calB(t-)} \frac{1}{|V_i^N|} \sum_{j\in V_i^N} \indc{j\in \calS(t-)} \\
& = \sum_{i \in \calA(t-)} \frac{1}{|V_i^N|} \sum_{j\in V_i^N} \indc{j\in \calS(t-)} -\sum_{i \in \calS(t-)} \frac{1}{|V_i^N|} \sum_{j\in V_i^N} \indc{j\in \calA(t-)}.
\end{align}
%
Thus, we further bound the downward drift in \eqref{eq: aaa} as
\begin{align*}
\lambda p_2 \sum_{i\in \calS(t-)} \frac{1}{|V_i^N|} \sum_{j \in V_i^N} \indc{j\in \calB(t-)}
& = \lambda p_2 \sum_{i\in \calB(t-)} \frac{1}{|V_i^N|} \sum_{j\in V_i^N} \indc{j\in \calS(t-)}\\
& + \lambda p_2\pth{\sum_{i\in \calS(t-)} \frac{1}{|V_i^N|} \sum_{j \in V_i^N} \indc{j\in \calB(t-)} - \sum_{i\in \calB(t-)} \frac{1}{|V_i^N|} \sum_{j\in V_i^N} \indc{j\in \calS(t-)}} \\
& \overset{(a)}{=} \lambda p_2 \sum_{i\in \calB(t-)} \frac{1}{|V_i^N|} \sum_{j\in V_i^N} \indc{j\in \calS(t-)}\\
&+ \lambda p_2\pth{\sum_{i\in \calA(t-)} \frac{1}{|V_i^N|} \sum_{j \in V_i^N} \indc{j\in \calS(t-)} - \sum_{i\in \calS(t-)} \frac{1}{|V_i^N|} \sum_{j\in V_i^N} \indc{j\in \calA(t-)}} \\
& \leq \lambda p_2 \pth{\sum_{i\in \calB(t-)} \frac{1}{|V_i^N|} \sum_{j\in V_i^N} \indc{j\in \calS(t-)}
+ \sum_{i\in \calA(t-)} \frac{1}{|V_i^N|} \sum_{j \in V_i^N} \indc{j\in \calS(t-)}},
\end{align*}
where equality (a) follows from \eqref{eq: diff}.
Thus, the total rate of moving away from state is $Z_1^N(t-)$ is upper bounded as
\begin{align}
\label{eq: eee1}
\lambda p_1 \frac{\mu}{K} |\calA(t-)| + \lambda\pth{\pth{1-\mu} p_1 + p_2} \sum_{i\in \calA(t-)} \frac{1}{|V_i^N|} \sum_{j \in V_i^N} \indc{j\in \calS(t-)} 
+ \lambda (p_1+p_2)\sum_{i\in \calB(t-)} \frac{1}{|V_i^N|} \sum_{j \in V_i^N} \indc{j\in \calS(t-)}.
\end{align}
So we get
\begin{align*}
 \mathbb{P}\left(\left. Z_1^N(t) = Z_1^N(t-) +1 \,\,\right|\,\, Z_1^N(t)\not=Z_1^N(t-) \right)
 \ge  \frac{\text{Eq.} ~\eqref{eq: ccc} ~~ }{\text{Eq.} ~ \eqref{eq: eee1} ~~ }.
\end{align*}
Note that
\begin{align}
\label{eq: ob}
\frac{\mu}{K} |\calA(t-)|  =  \frac{\mu}{K}\sum_{i\in \calA(t-)} \frac{1}{|V_i^N|} \sum_{j\in V_i^N} 1
\geq \frac{\mu}{K} \sum_{i\in \calA(t-)} \frac{1}{|V_i^N|} \sum_{j\in V_i^N} \indc{j\in \calS(t-)} \geq 0.
\end{align}
Thus, 
\begin{align*}
\mathbb{P}\left(\left. Z_1^N(t) = Z_1^N(t-) +1 \,\,\right|\,\, Z_1^N(t)\not=Z_1^N(t-) \right) 
 & \ge \frac{\frac{\mu}{K}p_1+\pth{1-\mu} p_1}{\frac{\mu}{K}p_1+\pth{1-\mu} p_1 + p_2}.
\end{align*}

Putting all the cases together, we conclude Lemma \ref{lm: random walker}.
\end{proof}

\subsection{Proof of Lemma \ref{lm: initial wealth}}
\label{app: initial wealth}
For a given time $t_c$ and for each $i=1, \cdots, N$, let
\begin{align}
\label{def: bour rv}
L_i(t_c) = \bm{1}_{\sth{X_{i,0}^N(0) =1 ~ \& ~\text{agent $i$ wakes up only once in $[0, t_c]$ $\& ~ M_i(t_c)=1$}}},
\end{align}
where $M_i(t_c)$ is the memory of agent $i$ at time $t_c$. 
Since an agent wakes up whenever its Poisson clock ticks and the Poisson clocks are independent among agents, it holds that
$L_i(t_c), \forall ~ i=1, \cdots, N$ are independent.  In addition, by symmetry, $L_i(t_c), \forall ~ i=1, \cdots, N$ are identically distributed.
We have
\begin{align}
\label{lm: upper bound}
\expect{L_i(t_c) } \geq (t_c\lambda) \exp(-t_c\lambda) \frac{\mu}{K} c_0 p_1.
%
\end{align}
Choosing $t_c=\frac{1}{\lambda}$, by Hoeffding's inequality, we have
\begin{align*}
\mathbb{P}\left( \sum_{i=1}^N L_i \pth{\frac{1}{\lambda}} \leq (1-C)\frac{\mu c_0p_1}{e K}N \right) \leq \exp\pth{-2\pth{\frac{C\mu c_0p_1}{e K}}^2N}.
\end{align*}
In addition, we know $Z_1^N\pth{\frac{1}{\lambda}} \geq \sum_{i=1}^N L_i \pth{\frac{1}{\lambda}}$.
Thus,
\begin{align*}
\mathbb{P}\left( Z_1^N\pth{\frac{1}{\lambda}} \leq (1-C)\frac{\mu c_0p_1}{e K}N \right) \leq \exp\pth{-2\pth{\frac{C\mu c_0p_1}{e K}}^2N}.
\end{align*}

\section{Proof of Theorem \ref{thm:localConvergence}}
\label{app:localConvergence}

The proof follows the same line of argument as in \cite{budhiraja2017} and \cite{BHAMIDI2018}. Since
\begin{align*}
 \expect{\sup_{0\leq t \leq T} \sum_{k=0}^K \left(\bm{X}_{i, k}^N(t) - \bm{X}_{i, k}(t)\right)^2 } &\leq \sum_{k=0}^K \expect{ \sup_{0\leq t \leq T} \left(\bm{X}_{i, k}^N(t) - \bm{X}_{i, k}(t)\right)^2 },
\end{align*}
we just need to show that the right hand side converges to zero. For $k\geq 0$, using that $\bm{X}_{i,k}^N (0)  = \bm{X}_{i,k} (0)$, we have
\begin{align}
& \sup_{0\leq t \leq T} \left(\bm{X}_{i, k}^N(t) - \bm{X}_{i, k}(t)\right)^2  \nonumber \\
& \quad \leq 2 \sup_{0\leq t \leq T} \left( \int\limits_{[0,K)\times[0,t]} \pth{\indc{\bm{X}_{i,k}^N (s^-) =0} \indc{y \in C_{i,k}^{N, +} (s^-)} - \indc{\bm{X}_{i,k} (s^-) =0} \indc{y \in C_{i,k}^+(s^-)}} \mathcal{N}_i(dy,ds) \right)^2  \nonumber \\
&\qquad + 2 \sup_{0\leq t \leq T} \left( \int\limits_{[0,K)\times[0,t]} \pth{\indc{\bm{X}_{i,k}^N (s^-) =1} \indc{y \in C_{i,k}^{N, -} (s^-)} -  \indc{\bm{X}_{i,k} (s^-) =1} \indc{y \in C_{i,k}^-(s^-)} } \mathcal{N}_i(dy,ds) \right)^2  \nonumber
\end{align}
\begin{align}
& \leq 4 \sup_{0\leq t \leq T} \left( \int\limits_{[0,K)\times[0,t]} \pth{\indc{\bm{X}_{i,k}^N (s^-) =0} - \indc{\bm{X}_{i,k} (s^-) =0}} \indc{y \in C_{i,k}^{N, +} (s^-)} \mathcal{N}_i(dy,ds)\right)^2 \nonumber \\
 &\qquad + 4 \sup_{0\leq t \leq T} \left( \int\limits_{[0,K)\times[0,t]} \indc{\bm{X}_{i,k} (s^-) =0} \pth{ \indc{y \in C_{i,k}^{N, +} (s^-)} - \indc{y \in C_{i,k}^+(s^-)}} \mathcal{N}_i(dy, ds)\right)^2 \nonumber \\
 &\qquad + 4 \sup_{0\leq t \leq T} \left( \int\limits_{[0,K)\times[0,t]} \pth{\indc{\bm{X}_{i,k}^N (s^-) =1} - \indc{\bm{X}_{i,k} (s^-) =1} } \indc{y \in C_{i,k}^{N,-}(s^-)} \mathcal{N}_i(dy,ds)\right)^2 \nonumber \\
 &\qquad + 4 \sup_{0\leq t \leq T} \left( \int\limits_{[0,K)\times[0,t]} \indc{\bm{X}_{i,k} (s^-) =1} \pth{ \indc{y \in C_{i,k}^{N, -} (s^-)} - \indc{y \in C_{i,k}^-(s^-)}} \mathcal{N}_i(dy, ds)\right)^2. \label{eq:firstRHS}
\end{align}
We bound the terms above individually. For the first one, we have
\begin{align}
&\sup_{0\leq t\leq T}\left( \int\limits_{[0,K)\times[0,t]} \pth{\indc{\bm{X}_{i,k}^N (s^-) =0}  - \indc{\bm{X}_{i,k} (s^-) =0}} \indc{y \in C_{i,k}^{N, +}(s) } \mathcal{N}_i(dy, ds)\right)^2 \nonumber \\
& =\sup_{0\leq t\leq T}\pth{\int\limits_{[0,K)\times[0,t]} \pth{\indc{\bm{X}_{i,k}^N (s^-) =0}  - \indc{\bm{X}_{i,k} (s^-) =0}} \indc{y \in C_{i,k}^{N, +}(s) } \left(\mathcal{N}_i(dy, ds) - \lambda dyds + \lambda dyds\right)}^2 \nonumber \\
& \leq 2 \sup_{0\leq t\leq T}\pth{ \int\limits_{[0,K)\times[0,t]} \left(\indc{\bm{X}_{i,k}^N (s^-) =0}  - \indc{\bm{X}_{i,k} (s^-) =0} \right) \indc{y \in C_{i,k}^{N, +}(s^-) }\left(\mathcal{N}_i(dy, ds) - \lambda dyds\right) }^2 \nonumber \\
& \qquad\qquad\qquad\qquad + 2 \lambda^2 \sup_{0\leq t\leq T}\pth{\int\limits_{[0,K)\times[0,t]} \pth{\indc{\bm{X}_{i,k}^N (s) =0}  - \indc{\bm{X}_{i,k} (s) =0}} \indc{y \in C_{i,k}^{N, +}(s) }dy ds }^2. \label{eq:secondRHS}
\end{align}
Note that the process
\begin{align*}
\int\limits_{[0,K)\times[0,t]} \left(\indc{\bm{X}_{i,k}^N (s^-) =0}  - \indc{\bm{X}_{i,k} (s^-) =0} \right) \indc{y \in C_{i,k}^{N, +}(s^-) }\left(\mathcal{N}_i(dy, ds) - \lambda dyds\right)
\end{align*}
is a martingale (part 2 of Lemma 1.12 in \cite{Eberle2015}). Then, Doob's inequality yields
\begin{align*}
  &\expect{\sup_{0\leq t\leq T}\pth{ \int\limits_{[0,K)\times[0,t]} \left(\indc{\bm{X}_{i,k}^N (s^-) =0}  - \indc{\bm{X}_{i,k} (s^-) =0} \right) \indc{y \in C_{i,k}^{N, +}(s^-) }\left(\mathcal{N}_i(dy, ds) - \lambda dyds\right) }^2} \\
  &\qquad \leq 4\expect{\pth{ \int\limits_{[0,K)\times[0,T]} \left(\indc{\bm{X}_{i,k}^N (s^-) =0}  - \indc{\bm{X}_{i,k} (s^-) =0} \right) \indc{y \in C_{i,k}^{N, +}(s^-) }\left(\mathcal{N}_i(dy, ds) - \lambda dyds\right) }^2}
\end{align*}
Furthermore, we have (part 3 of Lemma 1.12 in \cite{Eberle2015})
\begin{align*}
  & \expect{\pth{ \int\limits_{[0,K)\times[0,T]} \left(\indc{\bm{X}_{i,k}^N (s^-) =0}  - \indc{\bm{X}_{i,k} (s^-) =0} \right) \indc{y \in C_{i,k}^{N, +}(s^-) }\left(\mathcal{N}_i(dy, ds) - \lambda dyds\right) }^2} \\
  &\qquad\qquad\qquad\qquad\qquad = \lambda \expect{ \int\limits_{[0,K)\times[0,T]} \left( \left(\indc{\bm{X}_{i,k}^N (s) =0}  - \indc{\bm{X}_{i,k} (s) =0} \right) \indc{y \in C_{i,k}^{N, +}(s) } \right)^2 dyds} \\
  &\qquad\qquad\qquad\qquad\qquad \leq \lambda \expect{ \int_{0}^T \left( \indc{\bm{X}_{i,k}^N (s) =0}  - \indc{\bm{X}_{i,k} (s) =0} \right)^2 ds} \\
  &\qquad\qquad\qquad\qquad\qquad = \lambda \expect{ \int_{0}^T \left( \bm{X}_{i,k}^N (s)  - \bm{X}_{i,k} (s) \right)^2 ds},
\end{align*}
where the inequality comes from the fact that $\int_0^K \indc{y \in C_{i,k}^{N, +}(s) } dy \leq 1$ for all $s\geq 0$. By Tonelli's theorem we have
\begin{align*}
\expect{ \int_{0}^T \left( \bm{X}_{i,k}^N (s)  - \bm{X}_{i,k} (s) \right)^2 ds} = \int_{0}^T \expect{ \left( \bm{X}_{i,k}^N (s)  - \bm{X}_{i,k} (s) \right)^2} ds,
\end{align*}
and thus the first term in Equation \eqref{eq:secondRHS} is upper bounded by
\begin{equation}\label{eq:firstBound}
8\lambda \int_{0}^T \expect{ \left( \bm{X}_{i,k}^N (s)  - \bm{X}_{i,k} (s) \right)^2} ds.
\end{equation}

On the other hand, since $\int_0^K \indc{y \in C_{i,k}^{N, +}(s) } dy \leq 1$ for all $s\geq 0$, we have
\begin{align*}
  &\expect{\sup_{0\leq t\leq T}\pth{\int\limits_{[0,K)\times[0,t]} \pth{\indc{\bm{X}_{i,k}^N (s) =0}  - \indc{\bm{X}_{i,k} (s) =0}} \indc{y \in C_{i,k}^{N, +}(s) }dy ds }^2} \\
  & \qquad\qquad\qquad \leq \expect{\sup_{0\leq t\leq T}\pth{\int\limits_{[0,K)\times[0,t]} \left| \indc{\bm{X}_{i,k}^N (s) =0}  - \indc{\bm{X}_{i,k} (s) =0} \right| \indc{y \in C_{i,k}^{N, +}(s) }dy ds }^2} \\
  & \qquad\qquad\qquad \leq \expect{\sup_{0\leq t\leq T}\pth{\int_0^t \left| \indc{\bm{X}_{i,k}^N (s) =0}  - \indc{\bm{X}_{i,k} (s) =0} \right| ds }^2}.
\end{align*}
Furthermore, Jensen's inequality yields
\begin{align}
  &\expect{\sup_{0\leq t\leq T}\pth{\int_0^t \left| \indc{\bm{X}_{i,k}^N (s) =0}  - \indc{\bm{X}_{i,k} (s) =0} \right| ds }^2} \nonumber \\
  & \qquad\qquad\qquad\qquad\qquad\qquad\qquad\qquad \leq \expect{\sup_{0\leq t\leq T} \int_0^t \left(\indc{\bm{X}_{i,k}^N (s) =0}  - \indc{\bm{X}_{i,k} (s) =0} \right)^2 ds } \nonumber \\
  & \qquad\qquad\qquad\qquad\qquad\qquad\qquad\qquad \leq \expect{ \int_{0}^T \left(\bm{X}_{i,k}^N (s) - \bm{X}_{i,k} (s)  \right)^2 ds } \nonumber \\
  & \qquad\qquad\qquad\qquad\qquad\qquad\qquad\qquad = \int_{0}^T \expect{ \left(\bm{X}_{i,k}^N (s) - \bm{X}_{i,k} (s)  \right)^2} ds , \nonumber
\end{align}
where we used Tonelli's theorem in the last equality. Thus, the second term in Equation \eqref{eq:secondRHS} is upper bounded by
\[ 2\lambda^2 \int_{0}^T \expect{ \left(\bm{X}_{i,k}^N (s) - \bm{X}_{i,k} (s)  \right)^2} ds. \]
Combining this with Equation \eqref{eq:firstBound}, we obtain that the first term in Equation \eqref{eq:firstRHS} is upper bounded by
\begin{align}
 & 8(4+\lambda)\lambda \int_{0}^T \expect{ \left(\bm{X}_{i,k}^N (s) - \bm{X}_{i,k} (s)  \right)^2} ds \nonumber \\
 &\qquad\qquad\qquad\qquad\qquad\qquad \leq 8(4+\lambda)\lambda \int_{0}^T \expect{ \sup_{0\leq u \leq s} \left(\bm{X}_{i,k}^N (u) - \bm{X}_{i,k} (u)  \right)^2} ds \nonumber \\
 &\qquad\qquad\qquad\qquad\qquad\qquad \leq 8(4+\lambda)\lambda \int_{0}^T \max\limits_{i\in[N]} \sum\limits_{j=0}^K \expect{ \sup_{0\leq u\leq s} \Big(\bm{X}^N_{i, j}(u) - \bm{X}_{i, j}(u)\Big)^2} ds. \label{eq:firstRHSBound}
\end{align}
Moreover, the same argument yields the same upper bound for the third term in Equation \eqref{eq:firstRHS}.\\

For the fourth term in Equation \eqref{eq:firstRHS}, an analogous argument yields the upper bound
\begin{equation}\label{eq:intermediateBound}
 8(4+\lambda)\lambda \int_0^T \expect{ \int\limits_{[0,K)} \left(\indc{y \in C_{i,k}^{N, -}(s) } - \indc{y \in C_{i,k}^{-}(s) } \right)^2 dy } ds.
\end{equation}
Let us define the sets
\begin{align*}
 \tilde{C}_{i,0}^{-}(t) &:= \bigcup_{j=1}^K \left[j-1,\,\, j-1+ \frac{\mu p_j}{K} + \left(1-\mu\right) \frac{1}{|V_i^N|} \sum\limits_{\ell\in V_i^N} p_j \indc{\bm{X}_{\ell, j} (t) =1} \right), \\
 \tilde{C}_{i,k}^{-}(t) &:= \bigcup_{\overset{j=1}{j\neq k}}^K \left[j-1,\,\, j-1+  \frac{1}{|V_i^N|} \sum\limits_{\ell\in V_i^N} p_j \indc{\bm{X}_{\ell, j} (t) =1} \right), \qquad \forall\, k\geq 1.
\end{align*}
We have
\begin{align}
  &\expect{ \int\limits_{[0,K)} \left(\indc{y \in C_{i,k}^{N, -}(s) } - \indc{y \in C_{i,k}^{-}(s) } \right)^2 dy } \nonumber \\
  &\qquad \leq \expect{ \int\limits_{[0,K)} \left(\indc{y \in C_{i,k}^{N, -}(s) } - \indc{y \in \tilde{C}_{i,k}^{-}(s) } + \indc{y \in \tilde{C}_{i,k}^{-}(s) } - \indc{y \in C_{i,k}^{-}(s) } \right)^2 dy} \nonumber \\
  &\qquad \leq 2 \expect{ \int\limits_{[0,K)} \left(\indc{y \in C_{i,k}^{N,-}(s)} - \indc{y \in \tilde{C}_{i,k}^{-}(s) } \right)^2 dy + \int\limits_{[0,K)} \left( \indc{y \in \tilde{C}_{i,k}^{-}(s) } - \indc{y \in C_{i,k}^{-}(s) } \right)^2 dy}. \label{eq:thirdRHS}
\end{align}
For the first term, we have
\begin{align}
  &\expect{ \int\limits_{[0,K)} \left(\indc{y \in C_{i,k}^{N,-}(s)} - \indc{y \in \tilde{C}_{i,k}^{-}(s) } \right)^2 dy} \nonumber \\
  & \qquad\qquad\qquad\qquad\qquad\qquad\qquad\qquad \leq \expect{ \sum\limits_{\overset{j=1}{j\neq k}}^K \frac{1}{|V_i^N|} \left| \sum\limits_{\ell=1}^{|V_i^N|} p_j \left( \indc{\bm{X}^N_{i_\ell, j}(s)=1} - \indc{\bm{X}_{i_\ell, j}(s)=1} \right) \right|} \nonumber \\
  & \qquad\qquad\qquad\qquad\qquad\qquad\qquad\qquad \leq \expect{ \sum\limits_{\overset{j=1}{j\neq k}}^K \frac{1}{|V_i^N|}  \sum\limits_{\ell=1}^{|V_i^N|} \left| \indc{\bm{X}^N_{i_\ell, j}(s)=1} - \indc{\bm{X}_{i_\ell, j}(s)=1} \right|} \nonumber \\
  & \qquad\qquad\qquad\qquad\qquad\qquad\qquad\qquad = \expect{ \sum\limits_{\overset{j=1}{j\neq k}}^K \frac{1}{|V_i^N|} \sum\limits_{\ell=1}^{|V_i^N|} \Big(\bm{X}^N_{i_\ell, j}(s) - \bm{X}_{i_\ell, j}(s)\Big)^2} \nonumber \\
  & \qquad\qquad\qquad\qquad\qquad\qquad\qquad\qquad \leq \max\limits_{i\in[N]} \expect{\sum\limits_{j=0}^K \Big(\bm{X}^N_{i, j}(s) - \bm{X}_{i, j}(s)\Big)^2}. \label{eq:secondBound}
\end{align}
For the second term, we have
\begin{align*}
  &\expect{ \int\limits_{[0,K)} \left( \indc{y \in \tilde{C}_{i,k}^{-}(s) } - \indc{y \in C_{i,k}^{-}(s) } \right)^2 dy} \\
   &\qquad\qquad\qquad\qquad\qquad\qquad\qquad\qquad = \expect{ \sum\limits_{\overset{j=1}{j\neq k}}^K \frac{1}{|V_i^N|} \left| \sum\limits_{\ell=1}^{|V_i^N|} p_j \Big( \indc{\bm{X}_{i_\ell, j}(s)=1} - \mathbb{P}\left( \bm{X}_{i, j} (s) =1 \right)\Big)\right|} \\
  &\qquad\qquad\qquad\qquad\qquad\qquad\qquad\qquad = \sum\limits_{\overset{j=1}{j\neq k}}^K \frac{1}{|V_i^N|} \expect{ \left| \sum\limits_{\ell=1}^{|V_i^N|} p_j \Big( \indc{\bm{X}_{i_\ell, j}(s)=1} - \mathbb{P}\left( \bm{X}_{i, j} (s) =1 \right)\Big)\right|} \\
  &\qquad\qquad\qquad\qquad\qquad\qquad\qquad\qquad \leq \sum\limits_{\overset{j=1}{j\neq k}}^K \frac{1}{|V_i^N|} \expect{ \left| \sum\limits_{\ell=1}^{|V_i^N|} \Big( \indc{\bm{X}_{i_\ell, j}(s)=1} - \mathbb{P}\left( \bm{X}_{i, j} (s) =1 \right)\Big)\right|}.
\end{align*}
By Jensen's inequality, we have
\begin{align*}
 &\sum\limits_{\overset{j=1}{j\neq k}}^K \frac{1}{|V_i^N|} \expect{ \left| \sum\limits_{\ell=1}^{|V_i^N|} \Big( \indc{\bm{X}_{i_\ell, j}(s)=1} - \mathbb{P}\left( \bm{X}_{i, j} (s) =1 \right)\Big)\right|} \\
 &\qquad\qquad\qquad\qquad\qquad\qquad \leq \sum\limits_{\overset{j=1}{j\neq k}}^K \frac{1}{|V_i^N|} \left( \expect{ \left( \sum\limits_{\ell=1}^{|V_i^N|} \Big( \indc{\bm{X}_{i_\ell, j}(s)=1} - \mathbb{P}\left( \bm{X}_{i, j} (s) =1 \right)\Big)\right)^2 } \right)^{\frac{1}{2}}.
\end{align*}
Since the initial conditions $\{\bm{X}_i(0)\}_{i\in[N]}$ are i.i.d., $\{\bm{X}_i(s)\}_{i\in[N]}$ are also i.i.d., for all $s\geq 0$. Thus
\begin{align*}
  &\sum\limits_{\overset{j=1}{j\neq k}}^K \frac{1}{|V_i^N|} \left( \expect{ \left( \sum\limits_{\ell=1}^{|V_i^N|} \Big( \indc{\bm{X}_{i_\ell, j}(s)=1} - \mathbb{P}\left( \bm{X}_{i, j} (s) =1 \right)\Big)\right)^2 } \right)^{\frac{1}{2}} \\
  &\qquad\qquad\qquad\qquad\qquad\qquad\qquad\qquad = \sum\limits_{\overset{j=1}{j\neq k}}^K \frac{1}{|V_i^N|} \left( \expect{  \sum\limits_{\ell=1}^{|V_i^N|} \left( \indc{\bm{X}_{i_\ell, j}(s)=1} - \mathbb{P}\left( \bm{X}_{i, j} (s) =1 \right)\right)^2 } \right)^{\frac{1}{2}} \\
  &\qquad\qquad\qquad\qquad\qquad\qquad\qquad\qquad = \sum\limits_{\overset{j=1}{j\neq k}}^K \frac{1}{|V_i^N|} \left( |V_i^N| \expect{  \left( \indc{\bm{X}_{1, j}(s)=1} - \mathbb{P}\left( \bm{X}_{1, j} (s) =1 \right)\right)^2 } \right)^{\frac{1}{2}} \\
  &\qquad\qquad\qquad\qquad\qquad\qquad\qquad\qquad = \sum\limits_{\overset{j=1}{j\neq k}}^K \frac{1}{\sqrt{|V_i^N|}} \left( \expect{  \left( \indc{\bm{X}_{1, j}(s)=1} - \mathbb{P}\left( \bm{X}_{1, j} (s) =1 \right)\right)^2 } \right)^{\frac{1}{2}} \\
  &\qquad\qquad\qquad\qquad\qquad\qquad\qquad\qquad \leq \frac{K-1}{\sqrt{|V_i^N|}}
\end{align*}
Combining this with equations \eqref{eq:intermediateBound}--\eqref{eq:secondBound} we obtain that the fourth term in Equation \eqref{eq:firstRHS} is upper bounded by
\begin{align}
  & 16(4+\lambda)\lambda \int_0^T \left( \max\limits_{i\in[N]} \sum\limits_{j=0}^K \expect{ \Big(\bm{X}^N_{i, j}(s) - \bm{X}_{i, j}(s)\Big)^2} + \frac{K-1}{\sqrt{|V_i^N|}} \right) ds \nonumber \\
  & \qquad\qquad\qquad \leq 16(4+\lambda)\lambda \int_0^T \left( \max\limits_{i\in[N]} \sum\limits_{j=0}^K \expect{ \sup_{0\leq u\leq s} \Big(\bm{X}^N_{i,j}(u) - \bm{X}_{i,j}(u)\Big)^2} + \frac{K-1}{\sqrt{|V_i^N|}} \right) ds. \label{eq:secondRHSBound}
\end{align}
Moreover, an analogous argument yields the upper bound
\begin{equation}
 16(4+\lambda)\lambda \int_0^T \left( \max\limits_{i\in[N]} \sum\limits_{j=0}^K \expect{ \sup_{0\leq u\leq s} \Big(\bm{X}^N_{i,j}(u) - \bm{X}_{i,j}(u)\Big)^2} + \frac{1}{\sqrt{|V_i^N|}} \right) ds \label{eq:thirdRHSBound}
\end{equation}
for the second term in Equation \eqref{eq:firstRHS}.\\

Finally, combining equations \eqref{eq:firstRHSBound}, \eqref{eq:secondRHSBound}, and \eqref{eq:thirdRHSBound}, and substituting in Equation \eqref{eq:firstRHS}, we have that
\begin{align*}
  &\expect{ \sup_{0\leq t \leq T} \abth{\bm{X}_{i, k}^N(t) - \bm{X}_{i, k}(t)}^2 } \\
   &\qquad\qquad \leq 48(4+\lambda)\lambda \int_0^T \left( \max\limits_{i\in[N]} \sum\limits_{j=0}^K \expect{ \sup_{0\leq u\leq s} \Big(\bm{X}^N_{i, j}(u) - \bm{X}_{i, j}(u)\Big)^2} \right) ds + \frac{16(4+\lambda)\lambda KT}{\sqrt{|V_i^N|}},
\end{align*}
for all $k\geq 0$. It follows that
\begin{align*}
  &\max\limits_{i\in[N]} \sum_{k=0}^K \expect{ \sup_{0\leq t \leq T} \abth{\bm{X}_{i, k}^N(t) - \bm{X}_{i, k}(t)}^2 } \\
   &\qquad\qquad \leq 48(4+\lambda)\lambda(K+1) \int_0^T \left( \max\limits_{i\in[N]} \sum\limits_{k=0}^K \expect{ \sup_{0\leq u\leq s} \Big(\bm{X}^N_{i, k}(u) - \bm{X}_{i, k}(u)\Big)^2} \right) ds \\
   &\qquad\qquad\qquad\qquad\qquad\qquad\qquad\qquad\qquad\qquad\qquad\qquad\qquad\qquad + \frac{16(4+\lambda)\lambda KT(K+1)}{\sqrt{D_{\min}^N}}.
\end{align*}
Applying Gronwall's inequality, we obtain
\[ \max\limits_{i\in[N]} \sum_{k=0}^K \expect{ \sup_{0\leq t \leq T} \abth{\bm{X}_{i, k}^N(t) - \bm{X}_{i, k}(t)}^2 } \leq \frac{16(4+\lambda)\lambda KT(K+1)}{\sqrt{D_{\min}^N}} \exp\Big( 48(4+\lambda)\lambda(K+1)T \Big), \]
which concludes the proof.

\section{Proof of Theorem \ref{thm:transient}}
\label{app: detailed proof of mean-field}
Let us define the function $\eta:[0,T]\to[0,1]^{K+1}$, such that
\begin{equation}\label{eq:deterministicLimit}
\eta_k(t) := \mathbb{E}\big[\bm{X}_{1,k}(t)\big],
\end{equation}
for all $t\in[0,T]$, and for all $k=0,\dots,K$, where $\bm{X}_1(\cdot)$ is as in \eqref{eq:limitProcess}, with initial distribution $y$. Recall that
\begin{equation*}
  Y^N_k(t) = \frac{1}{N} \sum\limits_{i=1}^N \bm{X}_{i,k}^N (t),
\end{equation*}
for all $k=0,\dots,N$. Since
\[ \left|\frac{1}{N} \sum\limits_{i=1}^N \bm{X}_{i, k}^N(t) - \eta_k(t)\right| \leq 1, \]
we have
\[\expect{\sup_{0\leq t \leq T} \sum_{k=0}^K \left(\frac{1}{N} \sum\limits_{i=1}^N \bm{X}_{i, k}^N(t) - \eta_k(t)\right)^2 } \leq \expect{\sup_{0\leq t \leq T} \sum_{k=0}^K \left|\frac{1}{N} \sum\limits_{i=1}^N \bm{X}_{i, k}^N(t) - \eta_k(t)\right| }. \]
Furthermore,
\begin{align}
  &\expect{\sup_{0\leq t \leq T} \sum_{k=0}^K \left|\frac{1}{N} \sum\limits_{i=1}^N \bm{X}_{i, k}^N(t) - \eta_k(t)\right| } \nonumber \\
  &\qquad = \expect{\sup_{0\leq t \leq T} \sum_{k=0}^K \left|\frac{1}{N} \sum\limits_{i=1}^N \bm{X}_{i, k}^N(t) -\bm{X}_{i, k}(t) + \bm{X}_{i, k}(t) - \eta_k(t)\right| } \nonumber \\
  &\qquad \leq \expect{\sup_{0\leq t \leq T} \sum_{k=0}^K \left|\frac{1}{N} \sum\limits_{i=1}^N \bm{X}_{i, k}^N(t) -\bm{X}_{i, k}(t) \right| + \left| \frac{1}{N} \sum\limits_{i=1}^N \bm{X}_{i, k}(t) - \eta_k(t)\right| } \nonumber \\
  &\qquad \leq \expect{\sup_{0\leq t \leq T} \sum_{k=0}^K \left|\frac{1}{N} \sum\limits_{i=1}^N \bm{X}_{i, k}^N(t) - \bm{X}_{i, k}(t)\right| } + \expect{\sup_{0\leq t \leq T} \sum_{k=0}^K \left|\frac{1}{N} \sum\limits_{i=1}^N \bm{X}_{i, k}(t) - \eta_k(t)\right| }. \label{eq:twoTerms}
\end{align}
We first show that the second term converges to zero. Since
\[ \sup_{0\leq t \leq T} \sum_{k=0}^K \left|\frac{1}{N} \sum\limits_{i=1}^N \bm{X}_{i, k}(t) - y_k(t)\right| \leq 2, \]
the Dominated Convergence Theorem implies that
\begin{align*}
   \lim\limits_{N\to\infty} \expect{\sup_{0\leq t \leq T} \sum_{k=0}^K \left|\frac{1}{N} \sum\limits_{i=1}^N \bm{X}_{i, k}(t) - \eta_k(t)\right| } = \expect{ \lim\limits_{N\to\infty} \sup_{0\leq t \leq T} \sum_{k=0}^K \left|\frac{1}{N} \sum\limits_{i=1}^N \bm{X}_{i, k}(t) - \eta_k(t)\right| }.
\end{align*}
Furthermore, since the initial conditions of the processes $\bm{X}_{i}(\cdot)$ are i.i.d., then the whole processes are also i.i.d. by construction. Combining this with the fact that
\begin{equation*}
   \mathbb{E}\big[\bm{X}_{i,k}(t)\big] = \eta_k(t),
\end{equation*}
for all $k=0,\dots,K$, $i\in[N]$, and $t\in[0,T]$, and using standard arguments like in \cite{TX12,GTZ17}, it can be shown that
\begin{align*}
  \lim\limits_{N\to\infty} \sup_{0\leq t \leq T} \sum_{k=0}^K \left|\frac{1}{N} \sum\limits_{i=1}^N \bm{X}_{i, k}(t) - \eta_k(t)\right|  = 0, \qquad a.s.,
\end{align*}
and thus
\begin{align*}
  \mathbb{E}\left[\lim\limits_{N\to\infty} \sup_{0\leq t \leq T} \sum_{k=0}^K \left|\frac{1}{N} \sum\limits_{i=1}^N \bm{X}_{i, k}(t) - \eta_k(t)\right| \right]  = 0.
\end{align*}
For the first term in Equation \eqref{eq:twoTerms}, we have
\begin{align*}
  \expect{\sup_{0\leq t \leq T} \sum_{k=0}^K \left|\frac{1}{N} \sum\limits_{i=1}^N \bm{X}_{i, k}^N(t) - \bm{X}_{i, k}(t)\right| } &\leq \expect{\sup_{0\leq t \leq T} \frac{1}{N} \sum\limits_{i=1}^N \sum_{k=0}^K \left| \bm{X}_{i, k}^N(t) - \bm{X}_{i, k}(t)\right| } \\
  &\leq \frac{1}{N} \sum\limits_{i=1}^N \expect{\sup_{0\leq t \leq T} \sum_{k=0}^K \left| \bm{X}_{i, k}^N(t) - \bm{X}_{i, k}(t)\right| } \\
  &\leq \max\limits_{i\in[N]} \expect{\sup_{0\leq t \leq T} \sum_{k=0}^K \left| \bm{X}_{i, k}^N(t) - \bm{X}_{i, k}(t)\right| }.
\end{align*}
Since
\[ \left| \bm{X}_{i, k}^N(t) - \bm{X}_{i, k}(t)\right| = \left( \bm{X}_{i, k}^N(t) - \bm{X}_{i, k}(t)\right)^2, \]
and
\[ \lim_{N\to\infty} D_{\min}^N = \infty, \]
Theorem \ref{thm:localConvergence} implies that
\[ \lim\limits_{N\to\infty} \max\limits_{i\in[N]} \expect{\sup_{0\leq t \leq T} \sum_{k=0}^K \left| \bm{X}_{i, k}^N(t) - \bm{X}_{i, k}(t)\right| } = 0, \]
and thus the first term in Equation \eqref{eq:twoTerms} also converges to $0$.

The rest of the proof is devoted to showing that $\eta$ is the solution of the ODE defined by equations \eqref{limit 0} and \eqref{limit 1}, with initial condition $q$. Since $\eta_k(0)=\mathbb{E}\big[\bm{X}_{1,k}^N(0)\big]=q_k$, we have that
\begin{align*}
  \eta_k(t) & = \mathbb{E}\left[\bm{X}_{1,k}^N(0)\right] +
\mathbb{E}\left[ \int\limits_{[0,K)\times[0,t]} \pth{\indc{\bm{X}_{1,k} (s^-) =0} \indc{y \in C_{1,k}^+(s^-)} -  \indc{\bm{X}_{1,k} (s^-) =1} \indc{ y \in C_{1,k}^-(s^-)} } \mathcal{N}_1(dy,ds) \right] \\
&= q_k + \mathbb{E}\left[ \int\limits_{[0,K)\times[0,t]} \pth{\indc{\bm{X}_{1,k} (s) =0} \indc{y \in C_{1,k}^+(s)} -  \indc{\bm{X}_{1,k} (s) =1} \indc{ y \in C_{1,k}^-(s)} } \lambda dyds \right] \\
& + \mathbb{E}\left[ \int\limits_{[0,K)\times[0,t]} \pth{\indc{\bm{X}_{1,k} (s^-) =0} \indc{y \in C_{1,k}^+(s^-)} -  \indc{\bm{X}_{1,k} (s^-) =1} \indc{ y \in C_{1,k}^-(s^-)} } \big(\mathcal{N}_1(dy,ds) - \lambda dyds \big)\right].
\end{align*}
Since the term inside the second expectation is a martingale that starts at $0$ (part 2 of Lemma 1.12 in \cite{stochasticAnalysis}), Doob's optional stopping theorem implies that its expectation is equal to $0$ as well. As a result, we have that
\begin{align*}
  \eta_k(t) & = q_k + \mathbb{E}\left[ \int\limits_{[0,K)\times[0,t]} \pth{\indc{\bm{X}_{1,k} (s) =0} \indc{y \in C_{1,k}^+(s)} -  \indc{\bm{X}_{1,k} (s) =1} \indc{ y \in C_{1,k}^-(s)} } \lambda dyds \right] \\
&= q_k + \lambda  \int\limits_{[0,K)\times[0,t]} \left( \mathbb{E}\left[ \indc{\bm{X}_{1,k} (s) =0} \indc{y \in C_{1,k}^+(s)} \right] - \mathbb{E}\left[ \indc{\bm{X}_{1,k} (s) =1} \indc{ y \in C_{1,k}^-(s)} \right] \right) dyds,
\end{align*}
where in the last equality we used Fubini's theorem. For $k=0$, we have
\begin{align}
  \eta_0(t) & = q_0 - \lambda \int\limits_{[0,K)\times[0,t]} \mathbb{E}\left[ \indc{\bm{X}_{1,0} (s) =1} \right] \indc{ y \in C_{1,0}^-(s)} dyds, \nonumber \\
  &= q_0 - \lambda \int_0^t \eta_0(s) \sum\limits_{j=1}^K p_j \left( \frac{\mu}{K} + (1-\mu) \eta_j(s) \right) ds. \label{eq:drift0}
\end{align}
For $k\geq 1$, we have
\begin{align}
  \eta_k(t) &= q_k + \lambda  \int\limits_{[0,K)\times[0,t]} \left( \mathbb{E}\left[ \left(1-\indc{\bm{X}_{1,k} (s) =1}\right) \indc{y \in C_{1,k}^+(s)} \right] - \mathbb{E}\left[ \indc{\bm{X}_{1,k} (s) =1} \right] \indc{ y \in C_{1,k}^-(s)} \right) dyds, \nonumber \\
  &\qquad = q_k +\lambda \int_0^t p_k\left(\frac{\mu}{K} \eta_0(s) +  \big(1-\mu\eta_0(s)\big) \eta_k(s) \right) - p_k \eta_k(s)^2 -  \eta_k(s) \sum\limits_{\overset{j=1}{j\neq k}}^K  p_j \eta_j(s) ds \nonumber \\
   &\qquad = q_k + \lambda \int_0^t \pth{(1-\mu)p_k \eta_0(s)+ \sum_{j=1}^K (p_k-p_j)\eta_j(s)}\eta_k(s) + \eta_0(s) \frac{\mu}{K} p_k ds, \label{eq:driftj}
\end{align}
where the last equality comes from simple algebraic manipulation of the expressions. Combining equations \eqref{eq:drift0} and \eqref{eq:driftj} we conclude that $\eta(\cdot)$ is the solution to the ODE defined by equations \eqref{limit 0} and \eqref{limit 1}, with initial condition $q$.

\section{Proof of Theorem \ref{thm:expConvergence}} \label{app:proofExpConvergence}

We only prove the second case. The first one is just a special case of the proof.\\

We start by bounding the convergence rate of $y_0$. Our first characterization may be loose. However, we can use this loose bound to more refined characterization of the convergence rate of $y_1$. From \eqref{limit 0}, we have
\begin{align}
\label{eq: con rate lb}
\dot{y_0}(t) &= - y_0(t) \lambda \frac{\mu}{K}\sum_{k^{\prime}=1}^K p_{k^{\prime}} - y_0(t) \lambda \sum_{k^{\prime}=1}^K (1-\mu) p_{k^{\prime}} y_{k^{\prime}}(t)] ~
 \leq  - y_0(t) \lambda \frac{\mu}{K}\sum_{k^{\prime}=1}^K p_{k^{\prime}}.
\end{align}
Then, Gronwall's inequality implies
\begin{align}
\label{eq: loose bound}
y_0(t) \leq y_0(0)\exp \sth{- \pth{\lambda \frac{\mu}{K}\sum_{k=1}^K p_{k} } t}.
\end{align}
Although the bound in \eqref{eq: loose bound} is only for one entry of $y$, it can help us to get a convergence rate for $y_1$. Note that, at time
\[ \bar{t}_c = \frac{\log \frac{1}{c}}{\lambda \frac{\mu}{K}\sum_{k=1}^K p_k}, \]
we have
\begin{align*}
y_0(\bar{t}_c) \leq y_0(0) c.
\end{align*}
By \eqref{limit 1}, we know that
\begin{align*}
y_1(\bar{t}_c) ~ \geq ~ \frac{y_0(0)(1-c)}{K}.
\end{align*}

\vskip \baselineskip

On the other hand, from \eqref{limit 0}, we have
\begin{align}
\label{rate: bound 1}
\dot{y_1}(t)&=  y_0(t) \lambda  \frac{\mu}{K} p_1 +  \lambda \pth{(1-\mu)p_1 y_0(t)+ \sum_{k=1}^K (p_1-p_{k})y_{k}(t)}y_1(t) \nonumber \\
&\geq y_0(t) y_1(t) \lambda  \frac{\mu}{K} p_1 + \lambda \pth{(1-\mu)p_1 y_0(t)+ (p_1-p_2) \sum_{k=2}^K y_{k}(t)}y_1(t) \nonumber \\
&= \lambda \pth{\left(1-\mu+\frac{\mu}{K}\right)p_1 y_0(t)+ (p_1-p_2) \sum_{k=2}^K y_{k}(t)}y_1(t) \nonumber \\
&\geq  \lambda \cdot \min\left\{\left(1-\mu+\frac{\mu}{K}\right)p_1, p_1-p_2 \right\} \pth{ y_0(t)+  \sum_{k=2}^K y_{k}(t)}y_1(t) \nonumber \\
&= \lambda \cdot \min\left\{\left(1-\mu+\frac{\mu}{K}\right)p_1, p_1-p_2 \right\}  \pth{1- y_1(t)}y_1(t).
\end{align}
Let us define
\[ R\triangleq \lambda \min\left\{\left(1-\mu+\frac{\mu}{K}\right)p_1, p_1-p_2 \right\}, \]
and let $z_1$ be an auxiliary ODE equation such that
\begin{align}\label{ode: aux 1}
  \dot{z_1} &  = R \pth{1 - z_1}z_1,
\end{align}
with
\begin{align}
\label{eq: boundary y1}
z_1(0) \triangleq  y_1(\bar{t}_c) ~ \geq ~ \frac{y_0(0)(1-c)}{K}.
\end{align}
It can be shown that for all $t\in [\bar{t}_c, \infty)$,
\begin{align}
\label{aux: dom 0}
y_1 (\bar{t}_c+t) \geq z_1(t).
\end{align}
%
%
%
Thus, the convergence rate of $z_1$ provides a lower bound of the convergence rate of the original ODE system.
Note that $z_1$ is an autonomous and separable. Thus, we have
\begin{align*}
z_1(t) &= 1- \frac{1}{\frac{z_1(0)}{1-z_1(0)} \exp \big(R t \big) +1}\\
 &\geq 1- \frac{1}{\frac{y_0(0)(1-c)}{K-y_0(0)(1-c)} \exp \big(R t \big) +1},
\end{align*}
where the last inequality follows from \eqref{eq: boundary y1}.
By \eqref{aux: dom 0}, we know that $y_1(\bar{t}_c+ t)\geq z_1(t)$. Therefore, we conclude that
\[ y_1(t) \geq 1- \frac{1}{\frac{y_0(0)(1-c)}{K-y_0(0)(1-c)} \exp \big[R (t-t_c) \big] +1}, \]
for all $t\geq \bar{t}_c$.

\end{document}